\documentclass[11pt]{article}
\usepackage{graphicx} 
\usepackage[margin=1in]{geometry}
\usepackage[T1]{fontenc}

\usepackage{amsmath}
\usepackage{amssymb, amsthm, amsfonts, tikz-cd, mathrsfs,mathtools,stmaryrd,enumitem, algorithmic,float,comment,tikz}
\usetikzlibrary{backgrounds}
\usetikzlibrary{decorations.pathreplacing}
\usepackage{fullpage}
\usepackage{xcolor}
\usepackage[ruled,vlined,linesnumbered]{algorithm2e}
\usepackage[font=small, width=0.85\textwidth]{caption}

\usepackage[colorlinks,citecolor=blue,linkcolor=blue,urlcolor=red]{hyperref}

\usepackage{thm-restate}

\usepackage[multiuser,inline,nomargin]{fixme}

\FXRegisterAuthor{n}{en}{\color{blue}Nicole}
\FXRegisterAuthor{s}{ess}{\color{red}Sam}

\usepackage[capitalize,nameinlink]{cleveref}
\crefname{theorem}{Theorem}{Theorems}
\crefname{lemma}{Lemma}{Lemmas}
\crefname{claim}{Claim}{Claims}
\crefname{prop}{Proposition}{Propositions}

\newtheorem{theorem}{Theorem}[section]
\newtheorem{lemma}[theorem]{Lemma}

\newtheorem{observation}[theorem]{Observation}

\newtheorem{claim}[theorem]{Claim}
\theoremstyle{definition}
\newtheorem{definition}[theorem]{Definition}

\newcommand{\R}{\mathbb{R}}

\newcommand{\floor}[1]{\lfloor#1\rfloor}

\newcommand{\eps}{\varepsilon}

\def\ShowAuthNotes{1}
\ifnum\ShowAuthNotes=1
\newcommand{\authnote}[2]{\textcolor{blue}{[{\footnotesize {\bf #1:} { {#2}}}]}}
\else
\newcommand{\authnote}[2]{}
\fi

\title{Improved Hardness-of-Approximation for Token Swapping}

\author{
Sam Hiken\thanks{\url{shiken@mit.edu}. Supported by Paglia Post-Baccalaureate Fellowship and NSF grant CNS-2150186}
\\MIT 
\and 
Nicole Wein\thanks{\url{nswein@umich.edu}. This work was initiated during the DIMACS REU program at Rutgers University. This author was supported by a grant to DIMACS from the Simons Foundation (820931). }\\University of Michigan}
\date{}

\begin{document}

\maketitle

\abstract{We study the token swapping problem, in which we are given a graph with an initial assignment of one distinct token to each vertex, and a final desired assignment (again with one token per vertex). The goal is to find the minimum length sequence of \emph{swaps} of adjacent tokens required to get from the initial to final assignment. 

The token swapping problem is known to be NP-complete. It is also known to have a polynomial-time 4-approximation algorithm. From the hardness-of-approximation side, it is known to be NP-hard to approximate with ratio better than 1001/1000. 

Our main result is an improvement of the approximation ratio of the lower bound: We show that it is NP-hard to approximate with ratio better than 14/13. 

We then turn our attention to the \emph{0/1-weighted} version, in which every token has a weight of either 0 or 1, and the cost of a swap is the sum of the weights of the two participating tokens. Unlike standard token swapping, no constant-factor approximation is known for this version, and we provide an explanation. We prove that 0/1-weighted token swapping is NP-hard to approximate with ratio better than $(1-\varepsilon) \ln(n)$ for any constant $\epsilon>0$. 

Lastly, we prove two barrier results for the standard (unweighted) token swapping problem. We show that one cannot beat the current best known approximation ratio of 4 using a large class of algorithms which includes all known algorithms, nor can one beat it using a common analysis framework.



\vfill

\pagenumbering{gobble}
\pagebreak
 \pagenumbering{arabic}
 


\section{Introduction}

In the \emph{token swapping} problem~\cite{aicholzer:2021, MR4541302, akers:1989, MR2431751, MR1137822, MR1705338, MR1334632, MR3349550, miltzow_et_al:LIPIcs.ESA.2016.66, MR3805577, cayley1849lxxvii,MR796304,MR1691876,MR3917574,yasui2015swapping,portier1990whitney}, we are given an undirected $n$-vertex graph $G=(V,E)$, $n$ distinct tokens, and two one-to-one assignments of tokens to vertices: a \emph{starting} assignment, and a \emph{target} assignment. A swap along an edge $(u,v) \in E$ switches the locations of the tokens on $u$ and $v$. The token swapping problem asks how many swaps are needed to arrive at the target configuration from the starting configuration.

The token swapping problem is a fundamental and well-studied problem, and one of the central problems in the area of \emph{reconfiguration} algorithms. It has also found relevance in a number disparate areas including network engineering~\cite{akers:1989}, robot motion planning~\cite{MR4036097,surynek2019multi}, and game theory~\cite{gourves2017object}. Token swapping (mainly its parallel variant~\cite{MR1285588,MR3710080,MR4036097,MR3917574,MR1666061}) also has an extensively studied application to qubit routing (e.g.~\cite{banerjee2022locality, siraichi2019qubit, molavi2022qubit,bapat2023advantages,MR3964104,MR4638397,sharma2023noise,MR4594484}).
Algorithms and heuristics for the token swapping problem have also undergone  experimental evaluation~\cite{MR4261033,surynek2018finding}.

The token swapping problem is NP-complete~\cite{miltzow_et_al:LIPIcs.ESA.2016.66}, even when the underlying graph is a tree~\cite{aicholzer:2021}. As a result, the research literature has focused on approximation algorithms. Miltzow, Narins, Okamoto, Rote, Thomas, and Uno gave a 4-approximation~\cite{miltzow_et_al:LIPIcs.ESA.2016.66}, which remains the best known. For the special case of trees, there is a 2-approximation, which was independently discovered 3 times using different algorithms~\cite{akers:1989,MR1334632,YAMANAKA:2015}. 

From the hardness side, the above work~\cite{miltzow_et_al:LIPIcs.ESA.2016.66} proved that token swapping is APX-hard (on general graphs). Thus, unless $\text{P}=\text{NP}$, token swapping does not admit a PTAS, and instead there exists some positive constant $c$ for which there is no $c$-approximation. Regarding this constant $c$, they state \emph{``We want to point out that a crude estimate
for the constant c in [the NP-hardness result] is $c \approx 1 + 1/
1000$. We do not believe that it is worth to
compute c exactly. Instead, we hope that future research might find reductions with better
constants.''} This question is the main focus of our work. 

We also consider the \emph{0/1-weighted} version of token swapping, in which each token has a weight of either 0 or 1, and the cost of a swap is the sum of the weights of the two tokens. Despite being studied in prior work~\cite{MR4541302, aicholzer:2021}, there is no known approximation algorithm for 0/1-weighted with any non-trivial approximation ratio. Furthermore, there is no known separation between 0/1-weighted token swapping and standard token swapping.

\subsection{Our Results}
Our main result is the first hardness of approximation result for token swapping with an explicit approximation ratio of ``reasonable'' magnitude. Specifically, we show that it is NP-hard to obtain  better than a $14/13$-approximation: 

\begin{restatable}{theorem}{mainthm}\label{main-theorem}
  For any constant $\eps> 0$, it is NP-hard to approximate token swapping on graphs within a factor of $14/13-\eps$.
\end{restatable}

Our result is via a reduction from the label cover problem. We note that our result does not rely on the Unique Games Conjecture, and is instead a gap-preserving reduction from a regime of the label cover problem known to be NP-complete.


For 0/1-weighted token swapping, our next result explains the aforementioned gap in the literature by showing a large separation between the 0/1-weighted and unweighted versions. Specficially, we show that it is NP-hard to obtain better than an $(\ln n)$-approximation for 0/1-weighted token swapping: 

\begin{restatable}{theorem}{thmwts} \label{thm:WTS}
    For any constant $\eps > 0$, it is NP-hard to approximate weighted token swapping with $\{0,1\}$ weights on $n$ vertices within a factor of $(1-\eps) \cdot \ln n$.
\end{restatable}

Our reduction uses a completely different technique from our main result, and is a much simpler reduction, from the set cover problem. This is notable in comparison to the known reductions in the token swapping literature, which tend to be quite complicated~\cite{aicholzer:2021,MR4541302,miltzow_et_al:LIPIcs.ESA.2016.66,MR3805577}.

Our final two results are barrier results regarding the algorithm and analysis techniques that could possibly achieve better than the current best-known approximation ratio of 4 (for standard unweighted token swapping). All known token swapping algorithms have the natural property of \emph{local optimality}: they never perform a swap that brings \emph{both} tokens farther from their destinations.\footnote{We compare the notion of a locally optimal algorithm to that of an \emph{$\ell$-straying} algorithm, introduced in prior work on token swapping for \emph{trees}~\cite{aicholzer:2021}. These two notions are incomparable: a swap sequence can have either property without having the other. However, any algorithm that solves token swapping on general graphs must generate an $\Omega(n)$-straying swap sequence: consider the example of a cycle where each token wants to shift over by 1. For this reason, we do not consider $\ell$-straying algorithms for general graphs, and focus instead on locally optimal algorithms.} We show that, strangely, this property must be violated to achieve any approximation ratio better than 4. 

\begin{restatable}{theorem}{local} \label{thm:local-optimal}
    For any $\delta > 0$, there exists a token swapping instance $K$ so that for any locally optimal swap sequence of length $k$, $(4 - \delta) \cdot OPT(K) \leq k$.
\end{restatable}

In our second barrier result, we show that a proof technique used by all known analyses of approximation algorithms for token swapping, cannot yield better than a 4-approximation. See~\cref{sec:bar2} for more details.

Lastly, in the appendix we provide an alternative algorithm that gives a 4-approximation for token swapping (in addition to the known algorithm~\cite{miltzow_et_al:LIPIcs.ESA.2016.66}). While the algorithm of~\cite{miltzow_et_al:LIPIcs.ESA.2016.66} is an extension of the 2-approximation ``happy swap algorithm'' for trees~\cite{akers:1989}, our algorithm is an extension of the 2-approximation ``cycle algorithm'' for trees~\cite{YAMANAKA:2015}.

\subsection{Additional Related Work}
The token swapping problem and its variants have been studied from many angles. Exponential-time algorithms and hardness under the Exponential Time Hypothesis (ETH) were studied in~\cite{miltzow_et_al:LIPIcs.ESA.2016.66}. Parameterized complexity was studied in~\cite{MR3805577}, where the authors show that token swapping is W[1]-hard when the parameter is the number of swaps, and show further hardness under ETH. Token swapping has also been studied on a variety of special classes of graphs including cycles~\cite{MR796304}, stars~\cite{portier1990whitney,MR1691876}, brooms~\cite{MR3917574,MR4541302,MR1705338}, complete bipartite graphs~\cite{MR3349550}, complete split graphs~\cite{yasui2015swapping}, and cliques~\cite{cayley1849lxxvii} (dating back to Cayley in 1849). \emph{Colored} token swapping has also been studied, where each token has a color and same-colored tokens are indistinguishable~\cite{miltzow_et_al:LIPIcs.ESA.2016.66,MR4541302, MR3805577,MR3917573,MR3917573}. There are also several other models of token movement on graphs such as token sliding, token rotation, and token permutation
 (see e.g.~\cite{surynek2019multi}).
 See also the introductions of~\cite{MR4541302} and~\cite{aicholzer:2021} for a more details on the wealth of related work.



\section{Preliminaries}

\begin{definition}
    A \textsc{Token Swapping} instance $K = (G, T, f_1, f_2)$ consists of a graph $G = (V,E)$, a set of tokens $T$ where $|T| = |V|$, and two one-to-one assignments $f_1, f_2 : T \rightarrow V$. We call $f_1$ and $f_2$ the \emph{starting} and \emph{target} configurations respectively. A swap along an edge $(u,v) \in E$ switches the locations of the tokens on $u$ and $v$. The token swapping problem asks how many swaps are needed to arrive at the target configuration from the starting configuration. 
\end{definition}

For a token swapping instance $K$, we denote by $OPT(K)$ the length of the shortest swap sequence. For a vertex $v$, we denote by $f_1^{-1}(v)$ and $f_2^{-1}(v)$ the tokens that begin and end on $v$. When a token $t$ lies on the vertex $v_1$ of the path $p = v_1, v_2, ..., v_k$, we use the phrase \emph{bubbling t across p} to denote the sequence of swaps $(v_1, v_2), (v_2, v_3), ..., (v_{k-1}, v_k)$.

We also consider the following weighted variant of \textsc{Token Swapping}:
\begin{restatable}[Weighted Token Swapping]{definition}{weightedTS} \label{def:weightedTS}
     An instance of \textsc{Weighted Token Swapping} $W = (G, T, w, f_1, f_2)$ consists of a graph $G = (V,E)$, together with a set of tokens $T$ of size $|V|$. The weight function $w: T \rightarrow \R^{\geq 0}$ assigns to each token a non-negative real weight. As in standard \textsc{Token Swapping}, the one-to-one functions $f_1, f_2: T \rightarrow V$ map each token to a unique vertex, giving the starting aand target configurations of $W$. The weight of an edge swap is the sum of the weights of the tokens being swapped. The weight of a swap sequence is the sum of the weights of the constitutive edge swaps. The output to $W$ is the weight of the lowest-weight swap sequence from the starting to the target configuration.
\end{restatable}

\section{Hardness for standard token swapping}\label{sec:standard}
In this section, we prove \cref{main-theorem}.
\mainthm* 
This improves on the lower bound presented in \cite{miltzow_et_al:LIPIcs.ESA.2016.66}, which is roughly $1001/1000$. We obtain our lower bound via a gap-preserving reduction from \textsc{Label-Cover}, defined below.

\begin{definition}[Label Cover]
    An instance of \textsc{Label-Cover} $\Phi = (X, Y, E, \Sigma, \Pi)$ consists of a bipartite graph with vertex set $X \cup Y$ and edge set $E$, together with a finite alphabet $\Sigma$, and a set of constraints $\Pi = \{\Pi_e : e \in E\}$. Each constraint is a function $\Pi_e: \Sigma \rightarrow \Sigma$. 
    
    A labelling of $X\cup Y$ is a function $\lambda: X \cup Y \rightarrow \Sigma$ which assigns a label to each vertex. For $x\in X$, $y\in Y$, a constraint $\Pi_{(x,y)}$ is satisfied if $\Pi_{(x,y)}(\lambda(x)) = \lambda(y)$. 
\end{definition} 

We denote by $OPT(\Phi)$ the maximum fraction of constraints in $\Phi$ which can be satisfied by any labelling. An instance of the promise problem $\textsc{GapLabel-Cover}_{1,\gamma}(\Sigma)$ consists of a label-cover instance $\Phi$ with alphabet $\Sigma$, together with the guarantee that either $OPT(\Phi) = 1$ or $OPT(\Phi) < \gamma$. Here, $\gamma$ is a positive constant close to $0$. The following is a seminal result in the theory of hardness of approximation.

\begin{theorem} \label{lab-cov-hardness}
    For any constant $\gamma>0$, there exists a sufficiently large constant $|\Sigma|$ (dependent only on $\gamma$) such that $\textsc{GapLabel-Cover}_{1,\gamma}(\Sigma)$ is NP-hard. 
\end{theorem}

Arora and Lund \cite{arora:1996} describe a reduction from \textsc{3-SAT} which, together with Raz's Parallel Repetition Theorem \cite{doi:10.1137/S0097539795280895}, implies that \cref{lab-cov-hardness} remains true even on instances where the underlying graph is regular. Moreover, we may take the degree to be at least any arbitrarily large constant by creating many copies of our graph $X \cup Y$ and adding a copy of the constraint $\Pi_{(x,y)}$ between every copy of $x$ and every copy of $y$. Note that in such an instance $|X| = |Y|$.
We use a gap-preserving reduction from label-cover to prove \cref{main-theorem}.

\subsection{Construction}

Our reduction maps a label cover instance $\Phi = (X, Y, E, \Sigma, \Pi)$, where the base graph has degree $d$, to a token swapping instance $K(\Phi) = (G, T, f_1, f_2)$ as follows.

We construct $G$ by building a gadget $\textbf{Gad}(v)$ for each $v \in X \cup Y$, as follows. We begin building $\textbf{Gad}(v)$ with a base vertex $\textbf{bas}(v)$. We add $d$ additional vertices adjacent to $\textbf{bas}(v)$, which we call \emph{assignment vertices}. To each assignment vertex in $\textbf{Gad}(v)$ we identify a unique edge $(v,w) \in E$ (where $E$ is the label cover edge set) incident to $v$; 
we call this assignment vertex $\textbf{asg}(v,w)$. We then create $|\Sigma|$ paths, each with $\textbf{bas}(v)$ as an endpoint. We call these paths \emph{label paths}. If $v \in X$, then each label path in $\textbf{Gad}(v)$ contains $d-1$ edges; if $v \in Y$, then each label path in $\textbf{Gad}(v)$ contains $2d-1$ edges. To each label path in $\textbf{Gad}(v)$ we associate a unique $\sigma \in \Sigma$, and denote this label path by $\textbf{lab}(v, \sigma)$.

If $v \in X$, then we call $\textbf{Gad}(v)$ a \emph{left gadget}; otherwise, we call $\textbf{Gad}(v)$ a \emph{right gadget}. We call an assignment vertex a left (resp. right) assignment vertex if it is in a left (resp. right) gadget. We denote the set of left gadgets by $\textbf{L}$ and the set of right gadgets by $\textbf{R}$.

Whenever it is the case that the label pair $(\sigma_x, \sigma_y)$ satisfies the constraint $\Pi_{(x,y)}$, we add a path of $d$ edges connecting the endpoint of $\textbf{lab}(x, \sigma_x)$ opposite $\textbf{bas}(x)$ with the endpoint of $\textbf{lab}(y, \sigma_y)$ opposite $\textbf{bas}(y)$. We call such a path a \emph{satisfaction path}, and denote it by $\textbf{sat}(x, \sigma_x, y, \sigma_y)$.

\begin{figure}
    \centerline{\includegraphics[scale=.45]{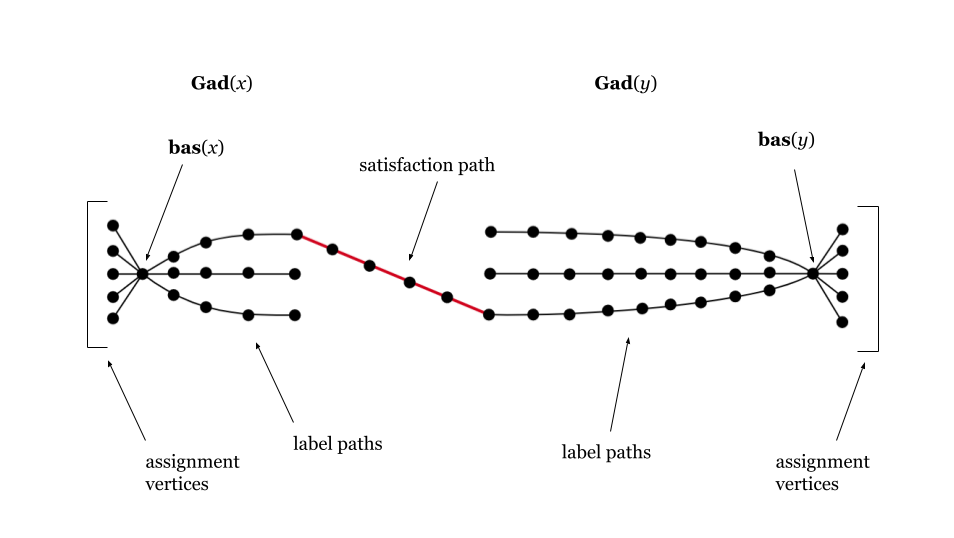}}
    \caption{A left and right gadget, joined by a satisfaction path; $d = 5$, $|\Sigma| = 3$. 
    }
    \label{fig:left-right-gadget}
\end{figure}

This gives us our construction for $G$. An illustration of a left and right gadget connected by a satisfaction path is given in \ref{fig:left-right-gadget}. We construct $f_1$ and $f_2$ as follows. For each $(x,y) \in E$, the tokens on $\textbf{asg}(x,y)$ and $\textbf{asg}(y,x)$ wish to swap positions with each other. That is, $f_2(f_1^{-1}(\textbf{asg}(x,y))) = \textbf{asg}(y,x)$, and $f_2(f_1^{-1}(\textbf{asg}(y,x))) = \textbf{asg}(x,y)$. 
We will refer to such a token as an \emph{assignment token}. We will call an assignment token $t$ where $f_1(t)$ is a left (resp. right) assignment vertex a left (resp. right) assignment token. For $t \in T$ that are not assignment tokens, $f_1(t) = f_2(t)$. That is, $t$ does not wish to move from its starting position.

The intuition behind the reduction is as follows: if there is a labelling $\lambda$ satisfying every constraint in $\Pi$, then all the assignment tokens that start in $\textbf{Gad}(v)$ can move to their destination along the single label path associated with $\lambda(v)$. Then, any assignment token seeking to \emph{enter} $\textbf{Gad}(v)$ can travel along this path, and will in the process swap with many of the assignment tokens leaving $\textbf{Gad}(v)$, bringing them closer to their destination. If, on the other hand, no good labelling exists, then having so many efficient swaps becomes impossible.

Now for the formal argument: we obtain our inapproximability lower bound via the following lemma.
\begin{lemma} \label{two-part-lemma}
  Let $\Phi = (X, Y, \Sigma, \Pi)$ be a label-cover instance whose graph is $d$-regular. The following two claims hold:
  \begin{itemize}
      \item  If $OPT(\Phi) = 1$, then $OPT(K(\Phi)) \leq \left( \frac{13}{4}d^2 + d \right) |X \cup Y |$.
      \item If $OPT(\Phi) < \gamma$, then $OPT(K(\Phi)) > \left( \frac{7 - \gamma}{2} d^2 - 5d\right)|X \cup Y |$.
  \end{itemize}
\end{lemma}

Given the label-cover instance $\Phi$, the token swapping instance $K(\Phi)$ can be computed in polynomial time. Because we can take $d$ to be an arbitrarily large constant and $\gamma$ an arbitrarily small constant, the hardness result of \cref{main-theorem} follows from \cref{two-part-lemma}. The next two sections are devoted to proving the first and second parts of \cref{two-part-lemma}, respectively.

\subsection{Completeness}
Let $\Phi$ be a label-cover instance as defined in \cref{two-part-lemma} such that $OPT(\Phi) = 1$. Let $\lambda: X \cup Y \rightarrow \Sigma$ be an optimal labelling. We prove the first half of \cref{two-part-lemma} by providing a swap sequence for $K(\Phi)$ using $(\frac{13}{4}d^2 + d)|X \cup Y|$ swaps.

We start by assigning an arbitrary order to the vertices in $X$, and an arbitrary order to the vertices in $Y$. We denote the ordered vertices $x_1, ..., x_{|X|} \in X$ and $y_1, ..., y_{|Y|} \in Y$. For each $x \in X$, the order on $Y$ induces an order on the assignment vertices of $\textbf{Gad}(x)$, given by $\textbf{asg}(x, y_{i_1}), \textbf{asg}(x, y_{i_2}), ..., \textbf{asg}(x, y_{i_d})$, where $i_1 < i_2 < ... < i_d$. We denote this order by $\textbf{asg}(x)_1, \textbf{asg}(x)_2, ..., \textbf{asg}(x)_d$. Similarly, for $y \in Y$, the order on $X$ induces an order on the assignment vertices of $\textbf{Gad}(y)$, which we denote by $\textbf{asg}(y)_1, \textbf{asg}(y)_2, ..., \textbf{asg}(y)_d$.

Our swap sequence proceeds in four stages, which we sketch here before a detailed description. In the first stage, we move the assignment tokens that begin on any left gadget $\textbf{Gad}(x)$ onto the label path associated with $\lambda(x)$, that is $\textbf{lab}(x, \lambda(x))$. This pushes the tokens which began on $\lambda(x)$ onto the assignment vertices in $\textbf{Gad}(x)$. The second phase rearranges the tokens now on $\textbf{Gad}(x)$'s assignment vertices so that at the end of Stage 4 they will end up on their target vertices (which were also their starting vertices). The third stage moves the assignment tokens that begin on any right gadget $\textbf{Gad}(y)$ onto $\textbf{lab}(x, \lambda(x))$. The fourth and final stage iterates over every $y \in Y$. A given iteration has two phases. In the first phase, we move each assignment token whose target vertex is in $\textbf{Gad}(y)$ to its target vertex, in the process pushing the assignment tokens in $\textbf{Gad}(y)$, which are now on $\textbf{lab}(y, \lambda(y))$, closer to their destination. In the second phase, we move each assignment token which started in $\textbf{Gad}(y)$ to its destination. The non-assignment tokens will finish the swap sequence on their target (that is, their starting) vertices.

Here is a formal description: 
\begin{enumerate}
    \item For $\textbf{Gad}(x) \in \textbf{L}$: 
        \begin{enumerate}
            \item For $k = 1, ..., d$: 
                \begin{enumerate}
                    \item Swap along the edge $(\textbf{asg}(x)_k, \textbf{bas}(x))$.
                    \item Bubble $f_1^{-1}(\textbf{asg}(x)_k)$ along the first $d-k$ edges of $\textbf{lab}(x, \lambda(x))$. 
                \end{enumerate}
        \end{enumerate}
    \item For $\textbf{Gad}(x) \in \textbf{L}$: 
        \begin{enumerate}
            \item For $k = 1, ..., \floor{d/2}$: 
                \begin{enumerate}
                    \item Swap along the edge $(\textbf{asg}(x)_k, \textbf{bas}(x))$.
                    \item Swap along the edge $(\textbf{asg}(x)_{d-k}, \textbf{bas}(x))$.
                    \item Swap along the edge $(\textbf{asg}(x)_k, \textbf{bas}(x))$.
                \end{enumerate}
        \end{enumerate}
    \item For $\textbf{Gad}(y) \in \textbf{R}$: 
        \begin{enumerate}
            \item For $k = 1, ..., d$:
            \begin{enumerate}
                \item Swap along the edge $(\textbf{bas}(y), \textbf{asg}(y)_k)$.
                \item Bubble $f_1^{-1}(\textbf{asg}(y)_k)$ along the first $d-k$ edges of $\textbf{lab}(y, \lambda(y))$.
            \end{enumerate}
        \end{enumerate}
    \item For $\ell = 1, ..., |Y|$:
    \begin{enumerate}
        \item For $k = 1, ..., d$: 
        \begin{enumerate}
            \item Let $\tilde{t}$ be the token with destination $\textbf{asg}(y_{\ell})_{d-k+1}$. Let $\tilde{x} \in X$ be chosen so that $\tilde{t}$ starts in $\textbf{Gad}(\tilde{x})$ (in particular, it will have started on $\textbf{asg}(\tilde{x},y_{\ell})$). 
            We will claim (see: \cref{claim:4(a)i}) that $\tilde{t}$ currently lies on the endpoint of $\textbf{lab}(\tilde{x}, \lambda(\tilde{x}))$ opposite $\textbf{bas}(\tilde{x})$. 
            Bubble $\tilde{t}$ along the $d$ edges of the satisfaction path it is incident to, which is $\textbf{sat}(\tilde{x}, \lambda(x), y_{\ell}, \lambda(y_{\ell}))$.
            \item Bubble $\tilde{t}$ along the $2d-1$ edges of the right label path it is now incident to, which is $\textbf{lab}(y, \lambda(y))$.
            \item Swap along the edge which brings $\tilde{t}$ to its destination, which is $(\textbf{bas}(y_{\ell}), f_2^{-1}(\textbf{asg}(\tilde{x},y_{\ell}))$.
        \end{enumerate}
        \item For $k = 1, ..., d$:
        \begin{enumerate}
            \item We will claim (see: \cref{claim:4(b)i}) that $f_1^{-1}(\textbf{asg}(y_{\ell})_{k})$ is of distance $k-1$ from the endpoint of $\textbf{lab}(y_{\ell}, \lambda(y_{\ell}))$ opposite $\textbf{bas}(y_{\ell})$. 
            Bubble $f_1^{-1}(\textbf{asg}(y_{\ell})_{k})$ along the $k-1$ edges to this endpoint.
            \item Let $\Bar{x} \in X$ be the vertex such that $f_2^{-1}(\textbf{asg}(\Bar{x},y_{\ell})) = f_1^{-1}(\textbf{asg}(y_{\ell})_{k})$. 
            Bubble $f_1^{-1}(\textbf{asg}(y_{\ell})_{k})$ along the $d$ edges of $\textbf{sat}(\Bar{x}, \lambda(\Bar{x}), y_{\ell}, \lambda(y_{\ell}))$.
            \item Bubble $f_1^{-1}(\textbf{asg}(y_{\ell})_{k})$ along the $d-1$ edges of $\textbf{lab}(\Bar{x}, \lambda(\Bar{x}))$.
            \item Swap along the edge $(\textbf{bas}(\Bar{x}), \textbf{asg}(\Bar{x}, y_{\ell}))$. 
        \end{enumerate}
    \end{enumerate}
\end{enumerate}

\begin{claim} \label{claim:4(a)i}
    Consider the start of the $\ell$-th iteration, $k$-th sub-iteration of Stage 4, Phase (a). 
    Let $\tilde{t}$ and $\tilde{x}$ be as defined in this sub-iteration. At this point in the swap sequence, $\tilde{t}$ lies on the endpoint of $\textbf{lab}(\tilde{x}, \lambda(x))$ opposite $\textbf{bas}(\tilde{x})$.
\end{claim}
\begin{proof}
    Let $j$ be the integer so that $\ell$ is the $j$-th smallest integer for which there is a satisfaction path from $\textbf{Gad}(\tilde{x})$ to $\textbf{Gad}(y_{\ell})$. By our construction, this means the start vertex of $\tilde{t}$ is $\textbf{asg}(\tilde{x})_j$. Therefore, Step 1(a)ii moves $\tilde{t}$ along the first $d-j$ edges of $\textbf{lab}(\tilde{x}, \lambda(x))$ to the endpoint opposite $\textbf{bas}(\tilde{x})$. This means that, at the beginning of Step 4, $\tilde{t}$ is on $\textbf{lab}(\tilde{x}, \lambda(x))$, and is of distance $j-1$ from the endpoint opposite $\textbf{bas}(\tilde{x})$.

    There are $j-1$ iterations of Step 4 where Step 4(a)i considers a token from $\tilde{x}$. During each of these iterations, there is one sub-iteration of 4(b) where Step 4(b)iii bubbles a token from $\textbf{Gad}(y_{\ell})$ along $\textbf{lab}(\tilde{x}, \lambda(x))$ towards $\textbf{bas}(\tilde{x})$. In the process, this token will swap with $\tilde{t}$, bringing it one spot closer to the endpoint of $\textbf{lab}(\tilde{x}, \lambda(x))$. At the start of the iteration considering of 4 where $\tilde{t}$ is considered in 4(a), $\tilde{t}$ will be on the endpoint of $\textbf{lab}(\tilde{x}, \lambda(x))$ opposite $\textbf{bas}(\tilde{x})$.
\end{proof}

\begin{claim} \label{claim:4(b)i}
    Consider the start of the $\ell$-th iteration, $k$-th sub-iteration of Stage 4, Phase (b). At this point in the swap sequence, $f_1^{-1}(\textbf{asg}(y_{\ell})_{k})$ is of distance $k-1$ from the endpoint of $\textbf{lab}(y_{\ell}, \lambda(y_{\ell}))$ opposite $\textbf{bas}(y_{\ell})$.
\end{claim}
\begin{proof}
    In Step 3(a)ii, the token starting on $\textbf{asg}(y_{\ell})_k$ was bubbled along the first $d-k$ edges of $\textbf{lab}(y_{\ell},\lambda(y_{\ell}))$. As such, at the start of Step 4, $f_1^{-1}(\textbf{asg}(y_{\ell})_{k})$ is of distance $d+k$ from the endpoint of $\textbf{lab}(y_{\ell}, \lambda(y_{\ell}))$ opposite $\textbf{bas}(y_{\ell})$. In Step 4(a), all $d$ assignment tokens with destination in $\textbf{Gad}(y_{\ell})$ were bubbled along $\textbf{lab}(y_{\ell},\lambda(y_{\ell}))$, each one swapping with $f_1^{-1}(\textbf{asg}(y_{\ell})_{k})$ once. At the start of 4(b), then, $f_1^{-1}(\textbf{asg}(y_{\ell})_{k})$ is of distance $k-1$ from the endpoint of $\textbf{lab}(y_{\ell}, \lambda(y_{\ell}))$ opposite $\textbf{bas}(y_{\ell})$.
\end{proof}

Now, we finish the proof for the completeness case. 
\begin{claim}
    The above swap sequence brings every token in $T$ to its target vertex.
\end{claim}
\begin{proof}
Let $t \in T$. We consider six cases. In each case, we show that $t$ ends the swap sequence on $f_2(t)$.


\paragraph{Case 1: $f_1(t)$ is an assignment vertex in a left gadget $\textbf{Gad}(x)$.} It follows by our construction that $f_2(t)$ is an assignment vertex in a right gadget $\textbf{Gad}(y)$. Because $(\lambda(x), \lambda(y))$ satisfies $(x,y)$, there exists a satisfaction path $\textbf{sat}(x, \lambda(x), y, \lambda(y))$, so there is a path $p$ of $4d$ edges from $f_1(t)$ to $f_2(t)$ consisting of: the edge $(f_1(t), \textbf{bas}(x))$, the paths $\textbf{lab}(x,\lambda(x))$, $\textbf{sat}(x, \lambda(x), y, \lambda(y))$, and $\textbf{lab}(y,\lambda(y))$, and the edge $(\textbf{bas}(y), f_2(t))$. We argue that the above swap sequence moves $t$ along $p$.

Let $k_X$ be the index of $f_1(t)$ in the order of the assignment vertices of $\textbf{Gad}(x)$. Step 1 moves $t$ along $d-k_X + 1$ edges of $p$. Step 2 only affects $t$ if $k_X = d$, in which case $t$ both begins and ends Step 2 on $\textbf{bas}(x)$. Step 3 does not affect $t$.

Because the order on $\textbf{Gad}(x)$'s assignment vertices is induced by the order on $Y$, there are $k_X-1$ iterations of Step 4 where there exists $(x,y_{\ell}) \in E$ that occur before the iteration where $y_{\ell} = y$. In each of these iterations, Step 4(b) has a sub-iteration where $\Bar{x} = x$, and 4(b)iii bubbles a token along $\textbf{lab}(x, \lambda(x))$, bringing $t$ one edge closer to $f_2(t)$.

Therefore, in the iteration of Step 4 where $y_{\ell} = y$, $t$ begins Step 4(a) on the endpoint of $\textbf{lab}(x, \lambda(x))$ opposite $\textbf{bas}(x)$. There is a sub-iteration of 4(a) where $\tilde{x} = x$. During this sub-iteration, 4(a)i bubbles $t$ along the $d$ edges of $\textbf{sat}(x, \lambda(x), y, \lambda(y))$. Step 4(a)ii bubbles $t$ along the $2d-1$ edges of $\textbf{lab}(y,\lambda(y))$, and 4(a)iii brings $t$ to $f_2(t)$. All the remaining steps do not affect $t$.

\paragraph{Case 2: $f_1(t)$ is an assignment vertex in a right gadget $\textbf{Gad}(y)$.} By our construction, $f_2(t)$ is an assignment vertex in a left gadget $\textbf{Gad}(x)$. Similarly to Case (1), there is a path $p$ of length $4d$ from $f_1(t)$ to $f_2(t)$ consisting of the edge $(f_1(t), \textbf{bas}(y))$, the paths $\textbf{lab}(y,\lambda(y))$, $\textbf{sat}(x, \lambda(x), y, \lambda(y))$, and $\textbf{lab}(x,\lambda(x))$, and the edge $(\textbf{bas}(x), f_2(t))$.

Let $k_Y$ be the index of $f_1(t)$ in the order of the assignment vertices of $\textbf{Gad}(y)$. Steps 1 and 2 do not affect $t$. Step 3 moves $t$ along $d-k_Y+1$ edges of $p$. There is one iteration of Step 4 where $y_{\ell} = y$. During this iteration, $t$ is swapped one edge along $p$ for each iteration of 4(a). Thus, at the beginning of 4(b), $t$ is distance $k_Y-1$ from the endpoint of $\textbf{lab}(y, \lambda(y))$ opposite $\textbf{bas}(y)$.

During the $k_Y$-th iteration of 4(b), $t$ is bubbled along $k_Y-1$ edges to the end of $\textbf{lab}(y, \lambda(y))$ in 4(b)i. During 4(b)ii, $\Bar{x} = x$, and $t$ is bubbled along the $d$ edges of $\textbf{sat}(x, \lambda(x), y, \lambda(y))$. During 4(b)iii, $t$ is brought to $\textbf{bas}(x)$, and in 4(b)iv $t$ is brought to $f_2(t)$.

\paragraph{Case 3: $f_1(t)$ is a vertex in a label path in a left gadget $\textbf{Gad}(x)$.} By our construction $f_2(t) = f_1(t)$. Let $j = \text{dist}(\textbf{bas}(x), f_1(t))$. Consider the iteration of Step 1 associated with $\textbf{Gad}(x)$. During the first $j$ sub-iterations of 1(a), $t$ is unaffected by 1(a)i, but 1(a)ii brings $t$ one edge closer to $\textbf{bas}(x)$. On the $j+1$-th sub-iteration of 1(a), $t$ begins 1(a)i on $\textbf{bas}(x)$, and 1(a)i brings $t$ to $\textbf{asg}(x)_{j+1}$. Step 2 brings $t$ to $\textbf{asg}(x)_{d-j}$. Step 3 does not affect $t$. 

There are $d$ distinct iterations of Step 4 where $\ell$ is such that $(x, y_{\ell}) \in E$. Because the order on the assignment vertices of $\textbf{Gad}(x)$ is induced by the order on $Y$, during the $i$-th of these $d$ iterations, it is the case that $f_2(f_1^{-1}(\textbf{asg}(x, y_{\ell}))) = \textbf{asg}(x)_{i}$. Therefore, on the $(d-j)$-th such iteration, $f_2(f_1^{-1}(\textbf{asg}(x, y_{\ell}))) = \textbf{asg}(x)_{d-j}$, and there is an iteration of Step 4(b) where 4(b)iv performs a swap along $(\textbf{bas}(x), \textbf{asg}(x)_{d-j})$. During each of the remaining $j$ iterations of 4 where $(x, y_{\ell}) \in E$, there is a sub-iteration of 4(b) where 4(b)iii swaps a token along $\textbf{lab}(x, \lambda(x))$, each reducing $t$'s distance from $f_1(t)$ by 1. Thus, $t$ concludes the swap sequence on its start vertex.

\paragraph{Case 4: $f_1(t)$ is a vertex in a label path in a right gadget $\textbf{Gad}(y)$, such that $0 \leq \text{dist}(f_1(t), \textbf{bas}(y)) \leq d-1$.} Let $j = \text{dist}(\textbf{bas}(x), f_1(t))$. Again, $f_2(t) = f_1(t)$. Steps 1 and 2 do not affect $t$. Consider the iteration of Step 3 associated with $\textbf{Gad}(y)$. During of the first $j$ sub-iterations of 3(a), Step 3(a)ii swaps a token with $t$, decreasing the distance from $t$ to $\textbf{bas}(y)$ by 1, until $t$ begins the $(j+1)$-th sub-iteration of 3(a) on $\textbf{bas}(y)$. Then, on the $(j+1)$-th sub-iteration of 3(a), Step 3(a)i brings $t$ onto $\textbf{asg}(y)_{d-j}$.

Consider the iteration of Step 4 where $y_{\ell} = y$. During the $(d-j)$-th sub-iteration of 4(a), Step 4(a)iii performs a swap along $(\textbf{bas}(y), \textbf{asg}(y)_{d-j})$, bringing $t$ to $\textbf{bas}(y)$. During each of the $j$ subsequent sub-iterations of 4(a), Step 4(a)ii brings $t$ one spot further from $\textbf{bas}(y)$ along $\textbf{lab}(y, \lambda(y))$. Step 4(b) does not affect $t$, nor do any subsequent iterations of 4, so $t$ ends the swap sequence on its target vertex.

\paragraph{Case 5: $f_1(t)$ is a vertex in a label path in a right gadget $\textbf{Gad}(y)$, such that $d \leq \text{dist}(f_1(t), \textbf{bas}(y)) \leq 2d-1$.} Let $j = \text{dist}(\textbf{bas}(x), f_1(t))$. Again, $f_2(t) = f_1(t)$. Steps 1-3 do not affect $t$. Consider the iteration of 4 where $y_{\ell} = y$. During each of the first $2d-1-j$ sub-iterations of 4(a), Step 4(a)ii performs a swap involving $t$, bringing $t$ one vertex closer to the endpoint of $\textbf{lab}(y, \lambda(y))$ opposite $\textbf{bas}(y)$. During the $(2d-j)$-th sub-iteration of 4(a), 4(a)ii swaps $t$ onto $\textbf{sat}(x, \lambda(x), y, \lambda(y))$, where $x \in X$ is the vertex such that $f_2(f_1^{-1}(\textbf{asg}(y)_{d-j}))$ is an assignment vertex in $\textbf{Gad}(x)$. All the remaining sub-iterations of 4(a) do not affect $t$.

During the $(j-d+1)$-th sub-iteration of 4(b), Step 4(b)ii swaps $t$ back onto $\textbf{lab}(y, \lambda(y))$. During each of the remaining $2d-j-1$ sub-iterations of 4(b), 4(b)ii brings $t$ one spot closer to $\textbf{bas}(y)$, so $t$ ends 4(b) at distance $j$ from $\textbf{bas}(y)$, on $f_1(t)$.

\paragraph{Case 6: $f_1(t)$ is a vertex in a satisfaction path $\textbf{sat}(x,\sigma_x, y, \sigma_y)$, and is not also an endpoint of a label path.} In this case, $f_2(t) = f_1(t)$. If $\lambda(x) \neq \sigma_x$ or $\lambda(y) \neq \sigma_y$, then $t$ is never affected by the swap sequence. Therefore, we assume that $\lambda(x) = \sigma_x$ and $\lambda(y) = \sigma_y$. Steps 1, 2, and 3 do not affect $t$. Step 4 only affects $t$ on the iteration where $y_{\ell} = y$. During this iteration, there is one sub-iteration of 3(a) where $\tilde{x} = x$. During this sub-iteration, Step 4(a)ii swaps $t$ with a neighboring token, but remains on the satisfaction path. Once Step 4(a) finishes executing, there is one sub-iteration of 4(b) where $\tilde{x} = x$. During this sub-iteration, Step 4(b)ii swaps $t$ back to $f_1(t)$.
\end{proof}

Now, we count the swaps in the sequence. For each $\textbf{Gad}(x) \in \textbf{L}$, Step 1 performs $\frac{d(d+1)}{2}$ swaps, resulting in a total of $\frac{d(d+1)}{4}|X \cup Y|$ swaps. For each $\textbf{Gad}(x)$, Step 2 performs at most $\frac{3d}{2}$ swaps, for a total of $\frac{3d}{4}|X \cup Y|$. Like Step 1, Step 3 performs a total of $\frac{d(d+1)}{4}|X \cup Y|$ swaps. Each time Step 4(a) is executed, it performs $3d$ swaps on each of $d$ assignment tokens. Each time Step 4(b) is executed, it performs $d(d-1)/2$ swaps on $\textbf{lab}(y, \lambda(y))$, as well as $2d$ swaps for each of $d$ right-outer tokens. In sum, Step 4 performs $3d^2+(d^2-d)/2+2d^2 = 11d^2/2 - d/2$ swaps on each iteration, and is executed $|Y| = \frac{1}{2}|X \cup Y|$ times, accounting for a total of $\frac{11}{4}d^2|X \cup Y| - \frac{1}{4}d|X \cup Y|$ swaps. In total, the procedure performs $\left(\frac{13}{4} d^2 + d\right)|X \cup Y|$ swaps.

\subsection{Soundness}
In this section, we will prove the second half of \cref{two-part-lemma}: if $OPT(\Phi) < \gamma$, then $OPT(K(\Phi)) > \left( \frac{7 - \gamma}{2} d^2 - 5d\right)|X \cup Y |$.

Let $\Phi$ be a label-cover instance such that $OPT(\Phi) < \gamma$, and let $S_{OPT}$ denote the optimal sequence of swaps on $K(\Phi)$. For a token $t \in T$, consider the subsequence of $S_{OPT}$ in which $t$ is one of the tokens being swapped. This subsequence traces out a walk in $G$. Consider the path from $f_1(t)$ to $f_2(t)$ obtained by removing all the closed sub-walks from this walk. We denote this path $\textbf{Swap}(t)$. 

We now partition the assignment tokens into two types: \emph{detour tokens} and \emph{non-detour} (assignment) tokens.

\begin{definition}[Detour token]
    A token $t$ is a \emph{detour token} if:
    \begin{itemize}
        \item $f_1(t)$ is an assignment vertex.
        \item $\textbf{Swap}(t)$ has more than $4d$ edges.
    \end{itemize}
\end{definition}

We will refer to all other assignment tokens as \emph{non-detour tokens}. This allows us to introduce the following notation: $det_L$ (resp. $det_R$) is the number of detour tokens $t$ such that $f_1(t)$ is a left (resp. right) assignment vertex. We will denote by $Det$ the set of detour tokens, and $Ndet$ the set of non-detour tokens.

We obtain lower bounds on two quantities: $B_{sat}$, the number of swaps occurring within satisfaction paths, and $B_{Gad}$, the number of swaps occurring within gadgets.
\begin{align}
    B_{sat} &\geq (d^2-3d)|X \cup Y| + \frac{d}{2}(det_L + det_R) \label{eqn:sat-bound} \\
    B_{Gad} &> \left( \frac{5 - \gamma}{2} d^2 -2d \right)|X \cup Y|  - \frac{d}{2}(det_L + det_R) \label{eqn:gad-bound}
\end{align}   

Combining these inequalities proves the second half of \cref{two-part-lemma}.

\subsubsection{Swaps on satisfaction paths}

First, we prove inequality \ref{eqn:sat-bound}.

\begin{observation}\label{observation:subpaths}
    If $t$ is a detour token, then $\textbf{Swap}(t)$ contains at least three satisfaction paths as subpaths.
\end{observation}

 \cref{observation:subpaths} follows from the construction of $G$; any path of length greater than $4d$ with a left assignment vertex and a right assignment vertex as endpoints traverses at least three satisfaction paths. 

Thus we can associate to each detour token $t$ a sequence of at least three satisfaction paths. We refer to the path(s) in this sequence which are not the first or last path as the \emph{intermediate path(s)} of $t$. Each detour token has at least one intermediate path.

\begin{claim} \label{claim:det-inter}
    $S_{OPT}$ contains at least $\frac{d}{2}(det_L + det_R)$ swaps where at least one token is a detour token visiting one of its intermediate paths.
\end{claim}
\begin{proof}
There are $det_L + det_R$ detour tokens, each of which participates in at least $d$ swaps on intermediate paths. Each swap involves at most $2$ tokens, hence the bound of $\frac{d}{2}(det_L + det_R)$. 
\end{proof}

Furthermore, let $ext(t)$ denote the number of swaps that $t$ participates in on satisfaction paths that are not intermediate paths. For $t \in Det$, $ext(t) \geq 2d$. For $t \in Ndet$, $ext(t) \geq d$.
Unfortunately, we cannot simply add these quantities together to obtain a lower bound on the number of swaps on satisfaction paths. This is because a swap may be between two different assignment tokens, so adding these quantities would risk double-counting. However, we can obtain an upper-bound on the amount of double-counting that can take place, using the following definition. 

\begin{definition}[Efficient satisfaction swap] \label{def:eff-sat}
    A swap between assignment tokens $t_1$ and $t_2$ is an \emph{efficient satisfaction swap} if one of the following is the case: 
    \begin{enumerate}
        \item neither $t_1$ nor $t_2$ is a detour token,
        \item one of $t_1$ or $t_2$ is not a detour token, the other is a detour token not on an intermediate path, and the swap is contained in $\textbf{Swap}(t)$, where $t$ is the detour token, or 
        \item $t_1$ and $t_2$ are detour tokens, neither is on an intermediate path, and the swap is contained in $\textbf{Swap}(t)$, where $t$ is $t_1$ or $t_2$.
    \end{enumerate}
\end{definition}
Let $k_{sat}$ denote the number of efficient satisfaction swaps in $S_{OPT}$. Combining \cref{def:eff-sat} with \cref{claim:det-inter}, we obtain the following inequality:  
\begin{equation*}
    B_{sat} \geq \frac{d}{2}\left( det_L + det_R \right) + \sum_{t \in Det} ext(t) + \sum_{t \in Ndet} ext(t) - k_{sat}
\end{equation*}
which implies that
\begin{equation} \label{ineq:k-sat}
B_{sat} \geq \frac{d}{2}\left( det_L + det_R \right) + 2d(det_L + det_R) + d \Big( d|X \cup Y| - (det_L + det_R) \Big) - k_{sat}.
\end{equation}
We give an upper bound on $k_{sat}$. We do so by proving individual upper bounds on each of the three types of efficient satisfaction swaps outlined in \cref{def:eff-sat}.

Swaps of type 1: \emph{neither $t_1$ not $t_2$ is a detour token}. Because the label-cover graph is simple, each edge $(x,y) \in E$ is mapped by our reduction to a unique pair of assignment vertices, those being $\textbf{asg}(x,y)$ and $\textbf{asg}(y,x)$. Therefore, any satisfaction path $\textbf{sat}(x, \sigma_x, y, \sigma_y)$ is traversed by at most two non-detour assignment tokens, those being $f_1^{-1}(\textbf{asg}(x,y)$ and $f_1^{-1}(\textbf{asg}(y,x)$. It follows that any non-detour token can swap on a satisfaction path with at most one other non-detour token 
Because there are $d|X \cup Y| - (det_L + det_R)$ non-detour assignment tokens, we obtain the following upper bound on swaps of type 1:
\[ \frac{1}{2}\Big( d|X \cup Y| - (det_L + det_R) \Big). \]
Swaps of type 2: \emph{one of $t_1$ or $t_2$ is not a detour token, and the other is a detour token not on an intermediate path}. Again, any satisfaction path $\textbf{sat}(x, \sigma_x, y, \sigma_y)$ is traversed by at most two non-detour assignment tokens, those being $f_1^{-1}(\textbf{asg}(x,y)$ and $f_1^{-1}(\textbf{asg}(y,x)$. A detour token moving across $\textbf{sat}(x, \sigma_x, y, \sigma_y)$ can swap with only one, depending on the direction in which it is moving. 
Because each detour token traverses exactly two non-intermediate satisfaction paths, the number of swaps of type 2 is at most:
\[2(det_L + det_R)\]
Swaps of type 3: \emph{$t_1$ and $t_2$ are detour tokens, but neither is on an intermediate path}. Each detour token participates in at most 
a combined $2d$ swaps on the first and last satisfaction paths that it visits. Each of these swaps involves at most $2$ such tokens, implying that the number of swaps of type 3 is at most:
\[ d(det_L + det_R). \]
Combining these bounds with inequality \ref{ineq:k-sat}, we are left with the following lower bound on $B_{sat}$:
\begin{align*}
    B_{sat} &\geq  \frac{d}{2}\left( det_L + det_R \right) + d \Big( d|X \cup Y| - (det_L + det_R) \Big) + 2d(det_L + det_R) \\
    &\quad -\frac{1}{2}\Big( d|X \cup Y| - (det_L + det_R) \Big) - 2(det_L + det_R) - d(det_L + det_R) \\
    &= d^2|X \cup Y| + \frac{3d}{2}(det_L + det_R) - \frac{d}{2}|X \cup Y| - \left( d + \frac{5}{2} \right)(det_L + det_R) \\
    &= \left( d^2 -\frac{d}{2} \right)|X \cup Y| + \left( \frac{d-5}{2} \right)(det_L + det_R) \\
    &\geq  \left( d^2 - 3d \right)|X \cup Y| + \frac{d}{2} (det_L + det_R)
\end{align*}
as desired. 

\subsubsection{Swaps on gadgets}

Next, we prove inequality \ref{eqn:gad-bound}:
\[ B_{Gad} > \left( \frac{5 - \gamma}{2} d^2 -2d \right)|X \cup Y|  - \frac{d}{2}(det_L + det_R). \]
Any assignment token $t$ participates in at least $3d-2$ swaps within label paths: $d-1$ swaps on a left label path and $2d-1$ swaps on a right label path. Moreover, we can identify the two label paths, one in a left gadget and one in a right gadget, which come at either end of $\textbf{Swap}(t)$. We refer to these paths as the \emph{first} and \emph{last legs} of $t$. Note that the first leg of $t$ comes in the gadget containing $t$'s start vertex, and the last leg comes in the gadget containing $t$'s target vertex.

A swap which corresponds to an edge in the first leg of one token $t_1$ and the last leg of another token $t_2$ we will call an \emph{efficient gadget swap}.
Note that such a swap appears in both $\textbf{Swap}(t_1)$ and $\textbf{Swap}(t_2)$. We denote the number of efficient gadget swaps in $S_{OPT}$ by $k_{Gad}$, and we denote the number of efficient gadget swaps occurring within the gadget $\textbf{Gad}(v)$ by $k_{Gad}(v)$. We observe that
\begin{equation} \label{ineq:k-gad}
    B_{Gad} \geq (3d^2-2d)|X \cup Y| - k_{Gad}
\end{equation}
and we seek an upper bound on $k_{Gad}$ accordingly.

To this end, we define the \emph{outflow} of a label path $p$ to be the number of tokens whose first leg is $p$, and the \emph{inflow} of $p$ to be the number of tokens whose last leg is $p$. For $\textbf{Gad}(v) \in \textbf{L} \cup \textbf{R}$, we call the label path in $\textbf{Gad}(v)$ with the maximum outflow (resp. inflow) the \emph{out-dominant} (resp. \emph{in-dominant}) label path of $\textbf{Gad}(v)$. We denote by $out(\textbf{Gad}(v))$ (resp. $in(\textbf{Gad}(v))$) 
the outflow (resp.~inflow) of the out-dominant (resp.~in-dominant) label path in $\textbf{Gad}(v)$. We obtain two inequalities.
\begin{claim} \label{claim:eff-bound}
    For any gadget $\textbf{Gad}(v)$, 
    \begin{itemize}
        \item $k_{Gad}(v) \leq d ( out(\textbf{Gad}(v)))$
        \item $k_{Gad}(v) \leq d ( in(\textbf{Gad}(v)))$.
    \end{itemize}
\end{claim}
\begin{proof}
    We prove the first inequality. Let $t \in T$ be an assignment token. Suppose $f_2(t)$ is in $\textbf{Gad}(v)$. Then $t$ participates in at most $out(\textbf{Gad}(v))$ efficient gadget swaps on  $\textbf{Gad}(v)$, as at most $out(\textbf{Gad}(v))$ such swaps can occur on the label path taken by $t$. 
    There are $d$ tokens $t$ where $f_2(t)$ is an assignment vertex in $\textbf{Gad}(v)$, hence the bound of $d ( out(\textbf{Gad}(v)))$. The proof of the second inequality is symmetric.
\end{proof}

\begin{claim} \label{claim:php}
    If $OPT(\Phi) < \gamma$, then
    \begin{align*}
         \sum_{x \in X} out(\textbf{Gad}(x)) + \sum_{y \in Y} in(\textbf{Gad}(y)) &< \left( \frac{1+ \gamma}{2} \right) d |X \cup Y| + det_L \\
         \sum_{x \in X} in(\textbf{Gad}(x)) + \sum_{y \in Y} out(\textbf{Gad}(y)) &< \left( \frac{1+ \gamma}{2} \right) d |X \cup Y| + det_R
    \end{align*}
\end{claim}
\begin{proof}
    We prove the first half of the claim; the proof of the second half is symmetric. Suppose for contradiction the first inequality does not hold. The left side of the inequality counts two quantities: (a) the total outflow from the out-dominant label path of each left gadget, plus (b) the total inflow to the in-dominant label path of each right gadget. Each left assignment token can contribute at most 1 to (a) and at most 1 to (b). Because $d |X \cup Y|$ is the total number of assignment tokens and the graph is regular (so $|X|=|Y|$), there are exactly $(d/2) |X \cup Y|$ left assignment tokens. As a result, even if the maximum number of left assignment tokens are contributing to only one of (a) or (b), at least $\frac{\gamma}{2}\left( d|X \cup Y| \right) + det_L$ contribute to both.
    Therefore, there are at least $\frac{\gamma}{2}\left( d|X \cup Y| \right)$ non-detour tokens starting in left gadgets which have an out-dominant label path as a first leg and an in-dominant label path as a last leg. Let $t$ be such a token. Because $t$ is a non-detour token, there is a satisfaction path between $t$'s first and last legs.
    By our construction, the constraint in the label cover instance associated with $t$ is satisfied by the labelling corresponding to $t$'s first and last legs. Consider a labelling of $\Phi$ which assigns to $x \in X$ the label corresponding to the out-dominant label path of $\textbf{Gad}(x)$, and to $y \in Y$ the label corresponding to the in-dominant label path of $\textbf{Gad}(y)$. There are at least $\gamma |E|$ constraints satisfied by this labelling, which is a contradiction.
\end{proof}
We average the inequalities in \cref{claim:eff-bound} and rearrange some terms to obtain the new bound for any $\textbf{Gad}(v)$: 
\[ \frac{k_{Gad}(v)}{d} \leq \frac{out(\textbf{Gad}(v)) + in(\textbf{Gad}(v))}{2}\]
and we average the two inequalities in \cref{claim:php} to obtain
\[  \sum_{v \in X \cup Y} \frac{out(\textbf{Gad}(v)) + in(\textbf{Gad}(v))}{2} < \left( \frac{1+ \gamma}{2} \right) d |X \cup Y| + \frac{det_L + det_R}{2}. \]
Combining these bounds, we arrive at the following inequality: 
\[ \sum_{v \in X \cup Y} \frac{k_{Gad}(v)}{d} <  \frac{1+ \gamma}{2}d |X \cup Y| + \frac{det_L + det_R}{2} \]
which implies 
\[ k_{Gad} <  \frac{1+ \gamma}{2} d^2 |X \cup Y| + \frac{d}{2}(det_L + det_R)  \]
which, combined with inequality \ref{ineq:k-gad}:
\[ B_{Gad} \geq (3d^2-2d)|X \cup Y| - k_{Gad}\]
implies inequality \ref{eqn:gad-bound}: 
\[B_{Gad} > \left( \frac{5 - \gamma}{2} d^2 -2d \right)|X \cup Y|  - \frac{d}{2}(det_L + det_R). \]

\section{Hardness for 0/1-Weighted Token Swapping}


In this section, we will prove the following theorem:

\thmwts*
We proceed via a reduction from the \textsc{Set Cover} problem, defined below.
\begin{definition}
    [Set Cover] An instance of \textsc{Set-Cover} $\Phi = (U,\{S_i\}_{1 \leq i \leq k} )$ consists of a finite universe set $U$, together with subsets $S_1, ..., S_k \subseteq U$. We seek to find the size of the smallest collection of subsets in $\{S_1, ..., S_k \}$ so that each element of $U$ is in at least one subset.
\end{definition}
Dinur and Steurer \cite{dinur:2014} prove the following theorem:
\begin{theorem}
    For any constant $\delta > 0$, it is NP-hard to approximate \textsc{Set Cover} within a factor of $(1-\delta) \cdot \ln ( |U| + k)$. 
\end{theorem}
We give a gap-preserving reduction from \textsc{Set Cover} to \textsc{Weighted Token Swapping}.

\paragraph{Construction} Given a \textsc{Set Cover} instance $\Phi = (U,\{S_i\}_{1 \leq i \leq k} )$, we produce a \textsc{Weighted Token Swapping} instance $K$ as follows. For each $u \in U$, we create two vertices, $v_u^1$ and $v_u^2$. We also create two tokens, each of weight 0, $t_u^1$ and $t_u^2$, so that $f_1(t_u^1) = v_u^2$ and $f_1(t_u^2) = v_u^1$, and $f_2(t_u^1) = v_u^1$ and $f_2(t_u^2) = v_u^2$.  That is, the tokens on $v_u^1$ and $v_u^2$ wish to swap with each other. For each $S_i \in \{S_i\}_{1 \leq i \leq k}$, we create a single vertex $v_i$. We also create a token $t_i$ of weight 1 such that $f_1(t_i) = f_2(t_i) = v_i$, that is, the token starting on $v_i$ wishes to stay there. Finally, for each $u \in S_i$, we add the edges $(v_u^1, v_i)$ and $(v_u^2, v_i)$. An example of such a construction is given in \cref{fig:weighted-gadget}.

\begin{figure}[h]
    \centerline{\includegraphics[scale=.28]{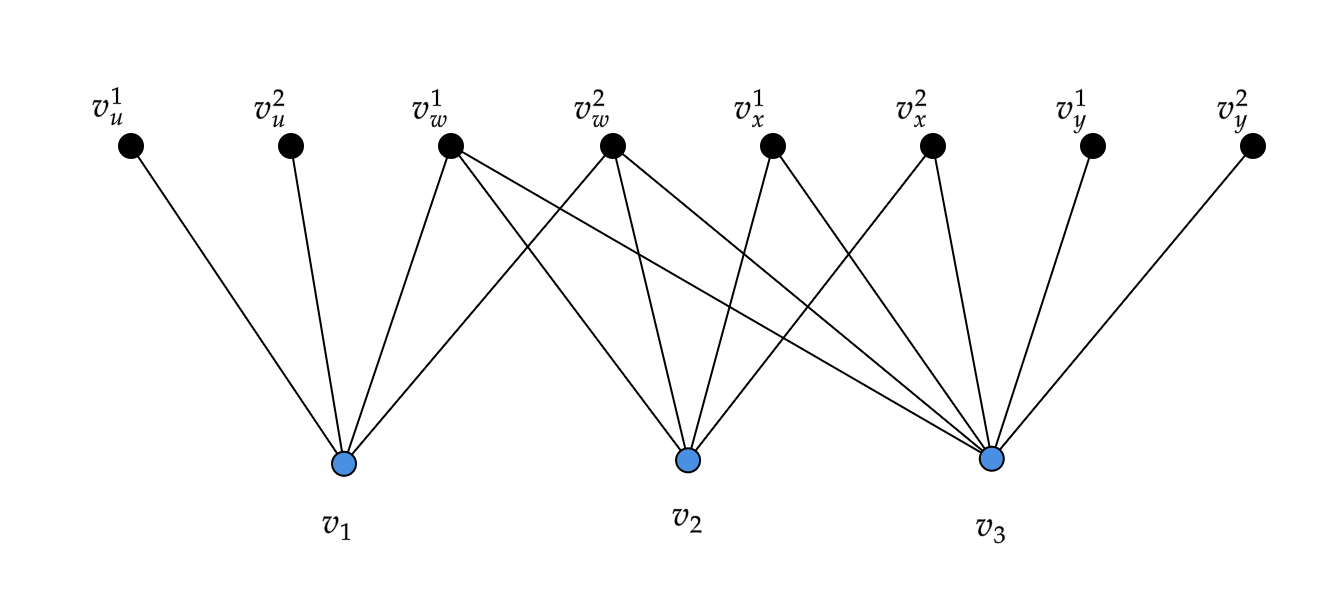}}
    \caption{The graph for a weighted token swapping instance constructed from a set cover instance with $U = \{u,w,x,y\}$ and $S_1 = \{u,w\}$, $S_2 = \{w,x\}$, and $S_3 = \{w,x,y\}$. Pairs of black vertices correspond to elements of $U$, while blue vertices correspond to sets. Two tokens from blue vertices must be displaced for tokens on black vertices to be moved to their destinations. 
    }
    \label{fig:weighted-gadget}
\end{figure}
We make the following claim.
\begin{claim} \label{claim:weighted-claim}
    $OPT(K) = 2 \cdot OPT(\Phi)$.
\end{claim}
Before giving the proof, we note that the above claim is enough to prove \cref{thm:WTS}. The number of vertices $n$ in the \textsc{Weighted Token Swapping} is $2|U| + k$. Therefore, to attain a reduction from $(1-\delta) \ln(|U|+k)$-approximate \textsc{Set Cover}, we set $\eps$ (a constant dependent on $\delta$) to be sufficiently small so that $(1-\eps) \ln(n) = (1-\eps) \ln(2|U| + k) \leq (1-\delta) \ln(|U|+k)$. The proof of \cref{claim:weighted-claim} is given below.
\begin{proof}
First, we show $OPT(K) \leq 2 \cdot OPT(\Phi)$. Let $F \subseteq \{S_i \}_{1 \leq i \leq k}$ be the smallest family of input subsets covering $U$. We perform the following sequence of swaps. For each $S_i \in F$, let $a \in S_i $ be an arbitrary element of $S_i$ which is not contained in any other set in $F$ (such an $a$ must exist, or we could remove $S_i$ from $F$). We perform the following sequence of swaps, which will bring $t_u^1$ and $t_u^2$ to their destinations, for any $u \in S_i$.
\begin{enumerate}
    \item Swap $t_a^1$ with $t_i$, so $t_a^1$ is on $v_i$.
    \item For any $u \in S_i$ so that $u \neq a$ and $t_u^1$ and $t_u^2$ have not already been moved to their target vertices:
    \begin{enumerate}
        \item Swap $t_u^1$ with $t_a^1$.
        \item Swap $t_u^1$ with $t_u^2$.
        \item Swap $t_u^2$ with $t_a^1$.
    \end{enumerate}
    This brings $t_u^1$ and $t_u^2$ to their destinations, and $t_a^1$ again on $v_i$. 
    \item Swap $t_a^1$ with $t_a^2$.
    \item Swap $t_a^2$ with $t_i$.
\end{enumerate}
This brings $t_u^1$ and $t_u^2$ to their destinations for every $u \in S_i$, and leaves $t_i$ on its destination vertex. Moreover, the only two steps in the above sub-sequence involving a weight-1 token are 1 and 4, so the total weight of the sequence is $2$. Because we repeat the sequence for every $S_i \in F$, the total weight of the optimal swap sequence is at most $2 \cdot |F| = 2 \cdot OPT(\Phi)$.

Now, we show $2 \cdot OPT(\Phi) \leq OPT(K)$. Let $T_{moved}$ be the set of weight-1 tokens moved from their starting vertex in the optimal swap sequence on $K$. Each such token contributes at least $2$ to the weight of the optimal swap sequence, as each token is moved from its start vertex, and during a later swap is moved back (because no two weight-1 tokens have adjacent start vertices, there is no double-counting between these swaps). For each $u \in U$, the tokens $t_u^1$ and $t_u^2$ want to switch positions with each other, but can only do so if an adjacent weight-1 token is displaced from its start vertex. A weight-1 token is only adjacent to $t_u^1$ and $t_u^2$, however, if it corresponds to an input subset $S_i$. Therefore, at least $OPT(\Phi)$ weight-1 tokens must be displaced from their start vertices, so $OPT(K) \geq 2 \cdot OPT(\Phi)$.
\end{proof}
\section{Additional barriers to improved algorithms}
In this section, we discuss two barriers to obtaining a $(4-\eps)$-approximation for unweighted token swapping on graphs. The first barrier is against a broad class of algorithms that encompasses all known approximation algorithms for token swapping. The second barrier shows that the technique used to prove approximation factors for existing algorithms cannot be used to prove an approximation factor of $4-\eps$ for any $\eps > 0$.
\subsection{A barrier against a general class of algorithms}
We define the following property for swap sequences on a \textsc{Token Swapping} instance. 
\begin{definition} [Local Optimality]
    A swap sequence is \emph{locally optimal} if every swap between tokens $t_1$ and $t_2$ does not move both $t_1$ and $t_2$ further from their destinations.
\end{definition}
All known algorithms for token swapping on trees or on general graphs work by returning the length of a swap sequence that is locally optimal. 

\local*
We construct a token swapping instance $K$ as follows. See~\cref{fig:outer-cycle}. Let $p$ and $q$ be parameters to be decided later. We can assume $p$ is even. We begin with a cycle $C^{out}$ of length $p \cdot q$. Starting with an arbitrary vertex, we label the vertices $v_0, ..., v_{p \cdot q - 1}$ in one direction (say, clockwise) around the cycle. We partition the vertices of $C^{out}$ into $p$ consecutive ``segments." For any $0 \leq i \leq p-1$, the vertices $\{v_{iq}, ... , v_{iq + (q-1)}\}$ form the $i$-th segment. If $i$ is even, we call such a segment an \emph{even segment}; else it is an \emph{odd segment}. For any $j$ so that $v_j$ is in an even segment, the token starting on $v_j$ has target vertex $v_{j+2q} \pmod{pq}$. If $v_j$ is in an odd segment, the token starting on $v_j$ has target vertex $v_{j-2q} \pmod{pq}$. That is, each token in an even segment wants to move to the corresponding vertex in the next even segment in the clockwise direction, while each token in an odd segment wants to move to the next odd segment over in the counter-clockwise direction. For each $0 \leq j \leq pq$, we add a path $P_j$ of $2q-2$ edges between $v_j$ and $v_{j + 2q}$. All the vertices on this path (except the endpoints in $C^{out}$) wish to stay on their start vertices. We call these paths \emph{inner paths}. We call the vertices in $C^{out}$ \emph{outer vertices}, and all other vertices \emph{inner vertices}. Tokens that begin on outer vertices (which also have destinations on outer vertices) we call \emph{outer tokens}; other tokens are \emph{inner tokens}. Outer tokens with start and target vertices in even segments we will call \emph{even tokens}; other outer tokens we will call \emph{odd tokens}. Note that each outer token begins on a vertex connected to its target vertex by an inner cycle. The outer cycle, together with a single inner cycle, is depicted in \cref{fig:outer-cycle}.

\begin{figure}[h]
    \centerline{\includegraphics[scale=.3]{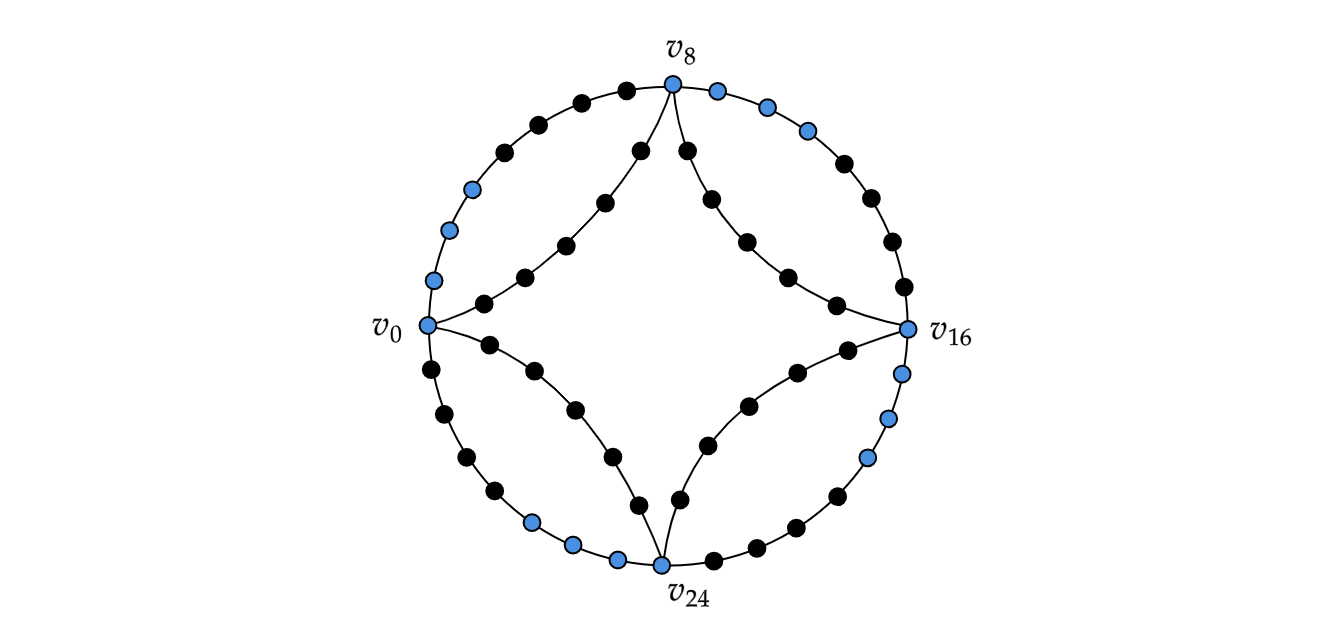}}
    \caption{An (incomplete) graph for a token swapping instance where $p = 8$, $q = 4$. The diagram depicts the outer cycle and one inner cycle, but leaves out the remaining inner cycles for legibility. In the outer cycle, vertices in even segments are colored blue, while those in odd segments are colored black. The token starting on $v_0$ has target $v_8$, the token starting on $v_8$ has target $v_{16}$, the token starting on $v_{16}$ has target $v_{24}$, and the token starting on $v_{24}$ has target $v_{0}$.}
    \label{fig:outer-cycle}
\end{figure}

\begin{claim}
    $OPT(K) \leq pq^2$.
\end{claim}
\begin{proof}
Here is a swap sequence bringing every token to its destination vertex:
\begin{enumerate}
    \item For $j = 0, ..., pq-1$:
    \begin{enumerate}
        \item If $f_1^{-1}(v_j)$ is an odd token, bubble it counter-clockwise around $C^{out}$ across $q$ edges.
    \end{enumerate}
    \item For $j = 0, ..., pq-1$:
    \begin{enumerate}
        \item If $f_1^{-1}(v_j)$ is an odd token, bubble it counter-clockwise around $C^{out}$ across $q$ edges. 
    \end{enumerate}
\end{enumerate}
We check that the above swap sequence brings every token to its target vertex. No swap occurs on an edge in an inner cycle, so none of the inner tokens are displaced. Because the vertices are considered in each loop in ascending order, whenever a token $f_1^{-1}(v_j)$ which began in an odd segment is bubbled $q$ edges counter-clockwise, the token $f_1^{-1}(v_{j-1})$ was just bubbled $q$ edges counter-clockwise itself during the previous iteration (if $v_{j-1}$ is in an odd segment). Therefore, $f_1^{-1}(v_j)$ only swaps with even tokens. Thus, each token which began in an odd segment winds up being bubbled counter-clockwise $2q$ times, and each token which began in an even segment is bubbled clockwise $2q$ times.
\end{proof}
\begin{claim} \label{claim:low-bound-barrier-1}
    Let $S$ be the shortest locally optimal swap sequence bringing every token to its destination, and let $k$ be the length of $S$. Then $4pq^2 - 5pq - 8q^2 \leq k$.
\end{claim}
We observe that for $0 \leq j \leq 2q-1$, the inner paths $P_j, P_{j + 2q}, P_{j+4q}, ..., P_{j+q(p-1)}$ form a cycle, which we will call the $j$-th \emph{inner cycle}, denoted $C^{in}_j$. Moreover, every token begins on a vertex in the same inner cycle as its target vertex. \cref{claim:low-bound-barrier-1} will follow from the next two claims.
\begin{claim} \label{claim:stays-on-inner-cycle}
    Let $a,b\in V$ be in the same inner cycle $C^{in}_j$, and let $P_{a,b}$ be a path from $a$ to $b$ containing at least one edge not in $C^{in}_j$. Then $\text{dist}(a,b) + 2 \leq |P_{a,b}|$, where $|P_{a,b}|$ is the number of edges in $P_{a,b}$.
\end{claim}
\begin{proof}
    We will prove the claim in two steps. First, we will modify $P_{a,b}$ to obtain a path $P_{a,b}'$ which is the same length as $P_{a,b}$ and which does not contain any edges from inner cycles other than $C^{in}_j$, but which contains at least one edge from $C^{out}$. Then, we will show that any path from $a$ to $b$ which contains edges from (and only from) both $C^{in}_j$ and $C^{out}$ has length at least $\text{dist}(a,b) + 2$.

    Suppose $P_{a,b}$ contains edges outside of $C^{in}_j$. If $P_{a,b}$ does not contain edges from a different inner cycle, then we set $P_{a,b}'$ to $P_{a,b}$ and move on to the argument in the next paragraph. Otherwise, $P_{a,b}$ contains an edge $e$ from a different inner cycle. Let $\ell$ be such that $C^{in}_{\ell}$ is the inner cycle containing $e$. Suppose $e$ is, in particular, the first edge to appear in $P_{a,b}$ which is in an inner cycle other than $C^{in}_j$. Therefore, $e$ is incident to a vertex $v_{\ell'}$ so that $\ell' \equiv \ell \pmod{2q}$, and $e$ is contained in either $P_{\ell'}$ or $P_{\ell'-2q}$. We can assume by symmetry that $e$ is contained in $P_{\ell'}$. Because $P_{a,b}$ is a shortest path, it is simple. By the construction of the graph then, $P_{a,b}$ contains the entirety of $P_{\ell'}$. Let $v_{j'}$ be the last vertex in $C^{in}_j$ preceding $v_{\ell'}$  in $P_{a,b}$. Because $e$ is incident to an outer vertex, the path between $v_{j'}$ and $v_{\ell'}$ only has edges in $C^{out}$. 
    The length of the sub-path of $P_{a,b}$ from the appearance of $v_{j'}$ to $v_{\ell' + 2q}$ is then $|\ell' - j'| + 2q-2$.  We replace this sub-path with the following path: the shortest path in $C^{in}_j$ from $v_{j'}$ to $v_{j' + 2q}$, then the shortest path in $C^{out}$ from $v_{j' + 2q}$ to $v_{\ell' + 2q}$. The length of this path is at most $|\ell' - j'| + 2q-2$, so we have not increased the length of $P_{a,b}$, but we have decreased the number of edges from an inner cycle other than $C^{in}_j$. We repeat this procedure until we obtain a path $P_{a,b}'$ which does not contain any edges from inner cycles other than $C^{in}_j$.

    We have a simple path $P_{a,b}'$ which contains edges only from $C^{in}_j$ and $C^{out}$, and at least one edge from $C^{out}$. Then there exists a consecutive sub-path in $P_{a,b}'$ of edges from $C^{out}$ of length $2q$, as this is the number of edges needed to go from one vertex in $C^{in}_j$ to another while traversing $C^{out}$. However, this sub-path can be replaced by the sub-path of length $2q-2$ between its endpoints using only edges from $C^{in}_j$. Therefore, $\text{dist}(a,b) + 2 \leq |P_{a,b}'| = |P_{a,b}|$.
\end{proof}
For $a,b$ in $C^{in}_j$, and for $c$ not in $C^{in}_j$ which is a neighbor of $a$, it is the case that $\text{dist}(a,b) < \text{dist}(c,b)$. Otherwise, there would be a path from $a$ to $b$ containing the edge $(a,c)$ of length at most $\text{dist}(a,b) + 1$, even though $(a,c)$ is not in $C^{in}_j$, which contradicts~\cref{claim:stays-on-inner-cycle}. Because each token starts in the same inner cycle as its target vertex, this implies that any swap across an edge in $C^{out}$ is not locally optimal. We proceed by proving a lower bound on the number of swaps needed to bring all the tokens in $C^{in}_j$ to their target vertex without leaving $C^{in}_j$.
\begin{claim} \label{claim:clockwise-stuff}
For any $j$, the number of swaps needed to bring all of the tokens on $C^{in}_j$ to their target vertices while swapping only along edges in $C^{in}_j$ is greater than $2pq - 5p/2 - 4q$.
\end{claim}
\begin{proof}
    Suppose without loss of generality that $v_j$ is in an even segment; by our construction, all the other outer vertices in $C^{in}_j$ are also in even segments. Moreover, each outer vertex in $C^{in}_j$ begins with a token which wants to move to the closest outer vertex in $C^{in}_j$ in the clockwise direction, of distance $2q-2$ away. All inner tokens want to stay on their start vertices. Note that $C^{in}_j$ has length $p(q-1)$, and contains $p/2$ outer vertices.
    
    In a given swap, one token in $C^{in}_j$ is moved clockwise, and the other is moved counter-clockwise. For a token $t$, let $clock(t)$ be the number of swaps in the optimal swap sequence on $C^{in}_j$ in which $t$ is moved clockwise, and $counter(t)$ the number of times its moved counter-clockwise. 

    Let $t_{in}$ be an inner token in $C^{in}_j$. Because $t_{in}$ has the same start and target vertex, $clock(t_{in}) - counter(t_{in}) \equiv 0 \pmod{p(q-1)}$. If $t_{out}$ is an outer vertex in $C^{in}_j$, then similarly $clock(t_{out}) - counter(t_{out}) \equiv 2q-2 \pmod{p(q-1)}$. Because each swap moves one token clockwise and one token counter-clockwise, $\sum_{t_{in} \in C^{in}_j} clock(t_{in}) - counter(t_{in}) = 0$. It follows that there is at least one token $t_{counter} \in C_{in}^j$ where $counter(t_{counter}) > clock(t_{counter})$ (else, each token in $C_{in}^j$ would end the swap sequence on its start vertex). Moreover, by our construction, the distance from $t_{counter}$'s start vertex to its target vertex in the counter-clockwise direction is at least $p(q-1) - 2q+2$, so $counter(t_{counter}) -clock(t_{counter}) \geq p(q-1) - 2q+2$. 
    
    If there is an additional token $t_{counter}'$ so that $p(q-1) - 2q + 2 \leq counter(t')$, then there are at least $2(p(q-1) - 2q) > 2pq - 5p/2 - 4q$ swaps in the optimal swap sequence, as desired.  Otherwise, there are at least $p/2-1$ outer tokens in $C^{in}_j$ which are not moved counter-clockwise at least $p(q-1) - 2q$ times. Each is involved in at least $2q-2$ swaps where it is the token being moved clockwise. Moreover, at most 1 of these swaps is with $t_{counter}$, because if a pair of tokens swap at least twice, then both swaps can be removed to yield an equivalent swap sequence. This results in a total of at least 
    \begin{align*}
        (p/2-1)(2q-3) + counter(t_{counter}) &\geq (p/2-1)(2q-3) + p(q-1) - 2q \\
        &> 2pq - 5p/2 - 4q
    \end{align*}
    swaps, as desired.
\end{proof}
Because there are $2q$ inner cycles, the total number of swaps taken by a sequence with only locally optimal swaps is at least $2q(2pq - 5p/2 - 4q) =  4pq^2 - 5pq -8q^2$, proving \cref{claim:low-bound-barrier-1}. We set $p,q$ to be large enough so that $4pq^2$ overtakes the smaller-order terms, so that for the parameter $\delta$,
\[ 4-\delta < (4pq^2 - 5pq -8q^2) / pq^2 < 4. \]
This proves \cref{thm:local-optimal}.
\subsection{A barrier against a proof technique}\label{sec:bar2}
All known approximation algorithms for token swapping on general graphs and trees \cite{miltzow_et_al:LIPIcs.ESA.2016.66} \cite{YAMANAKA:2015} \cite{akers:1989} use the following approach to bound the approximation factor. Given a \textsc{Token Swapping} instance $K$, consider the quantity $total(K)$, the sum across $K$'s tokens of the distance from the token's start vertex to its target vertex:
\[ total(K) \coloneq \sum_{t \in T} \text{dist}(f_1(t), f_2(t)). \]
We observe that $\frac{1}{2} total(K) \leq OPT(K)$, as each swap moves at most two tokens closer to their destinations. Proofs of correctness for approximation algorithms then show that they yield swap sequences of length at most $c \cdot total(K)$, which implies that the given algorithm is a $2c$-approximation. A natural approach to obtaining a $(4-\eps)$-approximation for token swapping on graphs, then, is to show that on any input an algorithm yields a swap sequence of length at most $c \cdot total(K)$, for some $c < 2$. Below, we show that this approach will not work.
\begin{theorem} \label{thm:apx-factor-barrier}
    For any $\delta > 0$, there exists a token swapping instance $K$ so that $OPT(K) \geq (2-\delta) \cdot total(K)$.
\end{theorem}
We give the following construction. Let $p,q$ be positive integer parameters to be specified later. Our \textsc{Token Swapping} instance $K$ consists of a cycle of length $p \cdot q$. Starting with an arbitrary vertex, we label the vertices $v_0, ..., v_{p \cdot q - 1}$ clockwise around the cycle. For each $0 \leq i \leq p-1$, the token starting on $v_{i \cdot q}$ has destination $v_{(i+1) \cdot q} \pmod{pq}$ (that is, it wants to move $q$ spots clockwise). All other tokens begin on their target vertices. 

In this construction, $total(K) = pq$. The next claim will allow us to prove \cref{thm:apx-factor-barrier}.
\begin{claim}
    $OPT(K) \geq 2pq - 2q - p + 1$.
\end{claim}
\begin{proof}
    The proof is similar to that for \cref{claim:clockwise-stuff}. For a token $t$, let $clock(t)$ be the number of times $t$ is involved in a swap that moves it clockwise, and $counter(t)$ then number of times it is moved counter-clockwise. Because each swap moves one token clockwise and one token counter-clockwise, $\sum_{t \in T} clock(t) - counter(t) = 0$. It follows that there is at least one token $t_{counter}$ where $counter(t_{counter}) > clock(t_{counter})$ (else, every token would end the swap sequence on its start vertex). Moreover, $counter(t_{counter})-clock(t_{counter}) \geq pq-q$, as by our construction the distance from each token's start vertex to target vertex is at least $pq-q$ in the counter-clockwise direction. 

    If there is an additional token $t_{counter}'$ so that $pq-q \leq counter(t_{counter}')$, then $t_{counter}$ and $t_{counter}'$ combined participate in at least $2pq-2q$ swaps, proving the claim.  Otherwise, there are at least $p-1$ tokens which have the property that they do not start on their target vertex, and they are not moved counter-clockwise at least $pq-q$ times, which, by our construction, means each needs to be moved clockwise at least $q$ times. Because each swap moves exactly one token clockwise, there are at least $(p-1)q$ swaps involving these tokens. Moreover, each such token swaps at most once with $t_{counter}$, because if a pair of tokens swap at least twice, then both swaps can be removed to yield an equivalent swap sequence. This results in a total of at least $pq-q + (p-1)q - (p-1) = 2pq - p - 2q + 1$ swaps, as desired.
\end{proof}
Thus, we can set $p,q$ to be large enough that the following inequality holds:
\[ (2qp - p - 2q + 1)/pq > (2 - \delta). \]
This proves \ref{thm:apx-factor-barrier}. Therefore, any $(4-\eps)$-approximation for \textsc{Token Swapping} would require a strategy for lower-bounding $OPT(K)$ that differs substantially from present techniques. 

\newpage
\bibliographystyle{plain}
\bibliography{ref.bib}

\begin{thebibliography}{10}

\bibitem{aicholzer:2021}
Oswin Aichholzer, Erik~D. Demaine, Matias Korman, Anna Lubiw, Jayson Lynch, Zuzana Mas\'{a}rov\'{a}, Mikhail Rudoy, Virginia Vassilevska~Williams, and Nicole Wein.
\newblock Hardness of token swapping on trees.
\newblock In {\em 30th annual {E}uropean {S}ymposium on {A}lgorithms}, volume 244 of {\em LIPIcs. Leibniz Int. Proc. Inform.}, pages Art. No. 3, 15. Schloss Dagstuhl. Leibniz-Zent. Inform., Wadern, 2022.

\bibitem{akers:1989}
S.B. Akers and B.~Krishnamurthy.
\newblock A group-theoretic model for symmetric interconnection networks.
\newblock {\em IEEE Transactions on Computers}, 38(4):555--566, 1989.

\bibitem{MR1285588}
Noga Alon, F.~R.~K. Chung, and R.~L. Graham.
\newblock Routing permutations on graphs via matchings.
\newblock {\em SIAM J. Discrete Math.}, 7(3):513--530, 1994.

\bibitem{arora:1996}
Sanjeev Arora and Carsten Lund.
\newblock Hardness of approximation.
\newblock In Dorit~S. Hochbaum, editor, {\em Approximation Algorithms for NP-Hard Problems}, chapter~10, pages 399--446. PWS Publishing, Boston, 2004.

\bibitem{banerjee2022locality}
Avah Banerjee, Xin Liang, and Rod Tohid.
\newblock Locality-aware qubit routing for the grid architecture.
\newblock In {\em 2022 IEEE International Parallel and Distributed Processing Symposium Workshops (IPDPSW)}, pages 607--613. IEEE, 2022.

\bibitem{MR3710080}
Indranil Banerjee and Dana Richards.
\newblock New results on routing via matchings on graphs.
\newblock In {\em Fundamentals of computation theory}, volume 10472 of {\em Lecture Notes in Comput. Sci.}, pages 69--81. Springer, Berlin, 2017.

\bibitem{bapat2023advantages}
Aniruddha Bapat, Andrew~M Childs, Alexey~V Gorshkov, and Eddie Schoute.
\newblock Advantages and limitations of quantum routing.
\newblock {\em PRX Quantum}, 4(1):010313, 2023.

\bibitem{MR4541302}
Ahmad Biniaz, Kshitij Jain, Anna Lubiw, Zuzana Mas\'{a}rov\'{a}, Tillmann Miltzow, Debajyoti Mondal, Anurag~Murty Naredla, Josef Tkadlec, and Alexi Turcotte.
\newblock Token swapping on trees.
\newblock {\em Discrete Math. Theor. Comput. Sci.}, 24(2):Paper No. 9, 37, 2022.

\bibitem{MR3805577}
\'{E}douard Bonnet, Tillmann Miltzow, and Pawe\l{} Rz\k{a}\.{z}ewski.
\newblock Complexity of token swapping and its variants.
\newblock {\em Algorithmica}, 80(9):2656--2682, 2018.

\bibitem{cayley1849lxxvii}
Arthur Cayley.
\newblock Lxxvii. note on the theory of permutations.
\newblock {\em The London, Edinburgh, and Dublin Philosophical Magazine and Journal of Science}, 34(232):527--529, 1849.

\bibitem{MR3964104}
Andrew~M. Childs, Eddie Schoute, and Cem~M. Unsal.
\newblock Circuit transformations for quantum architectures.
\newblock In {\em 14th {C}onference on the {T}heory of {Q}uantum {C}omputation, {C}ommunication and {C}ryptography}, volume 135 of {\em LIPIcs. Leibniz Int. Proc. Inform.}, pages Art. No. 3, 24. Schloss Dagstuhl. Leibniz-Zent. Inform., Wadern, 2019.

\bibitem{MR4036097}
Erik~D. Demaine, S\'{a}ndor~P. Fekete, Phillip Keldenich, Henk Meijer, and Christian Scheffer.
\newblock Coordinated motion planning: reconfiguring a swarm of labeled robots with bounded stretch.
\newblock {\em SIAM J. Comput.}, 48(6):1727--1762, 2019.

\bibitem{dinur:2014}
Irit Dinur and David Steurer.
\newblock Analytical approach to parallel repetition.
\newblock In {\em Proceedings of the Forty-Sixth Annual ACM Symposium on Theory of Computing}, STOC '14, page 624–633, New York, NY, USA, 2014. Association for Computing Machinery.

\bibitem{gourves2017object}
Laurent Gourv{\`e}s, Julien Lesca, and Ana{\"e}lle Wilczynski.
\newblock Object allocation via swaps along a social network.
\newblock In {\em 26th International Joint Conference on Artificial Intelligence (IJCAI’17)}, pages 213--219, 2017.

\bibitem{MR4638397}
Takehiro Ito, Naonori Kakimura, Naoyuki Kamiyama, Yusuke Kobayashi, and Yoshio Okamoto.
\newblock Algorithmic theory of qubit routing.
\newblock In {\em Algorithms and data structures}, volume 14079 of {\em Lecture Notes in Comput. Sci.}, pages 533--546. Springer, Cham, [2023] \copyright 2023.

\bibitem{MR796304}
Mark~R. Jerrum.
\newblock The complexity of finding minimum-length generator sequences.
\newblock {\em Theoret. Comput. Sci.}, 36(2-3):265--289, 1985.

\bibitem{MR3917574}
Jun Kawahara, Toshiki Saitoh, and Ryo Yoshinaka.
\newblock The time complexity of permutation routing via matching, token swapping and a variant.
\newblock {\em J. Graph Algorithms Appl.}, 23(1):29--70, 2019.

\bibitem{miltzow_et_al:LIPIcs.ESA.2016.66}
Tillmann Miltzow, Lothar Narins, Yoshio Okamoto, G\"{u}nter Rote, Antonis Thomas, and Takeaki Uno.
\newblock Approximation and hardness of token swapping.
\newblock In {\em 24th {A}nnual {E}uropean {S}ymposium on {A}lgorithms (ESA)}, volume~57 of {\em LIPIcs. Leibniz Int. Proc. Inform.}, pages Art. No. 66, 15. Schloss Dagstuhl. Leibniz-Zent. Inform., Wadern, 2016.

\bibitem{molavi2022qubit}
Abtin Molavi, Amanda Xu, Martin Diges, Lauren Pick, Swamit Tannu, and Aws Albarghouthi.
\newblock Qubit mapping and routing via maxsat.
\newblock In {\em 2022 55th IEEE/ACM international symposium on Microarchitecture (MICRO)}, pages 1078--1091. IEEE, 2022.

\bibitem{MR1691876}
Igor Pak.
\newblock Reduced decompositions of permutations in terms of star transpositions, generalized {C}atalan numbers and {$k$}-ary trees.
\newblock {\em Discrete Math.}, 204(1-3):329--335, 1999.

\bibitem{portier1990whitney}
Frederick~J Portier and Theresa~P Vaughan.
\newblock Whitney numbers of the second kind for the star poset.
\newblock {\em European Journal of Combinatorics}, 11(3):277--288, 1990.

\bibitem{doi:10.1137/S0097539795280895}
Ran Raz.
\newblock A parallel repetition theorem.
\newblock {\em SIAM Journal on Computing}, 27(3):763--803, 1998.

\bibitem{MR4261033}
Bruno Schmitt, Mathias Soeken, and Giovanni De~Micheli.
\newblock Symbolic algorithms for token swapping.
\newblock In {\em 2020 {IEEE} 50th {I}nternational {S}ymposium on {M}ultiple-{V}alued {L}ogic---{ISMVL} 2020}, pages 28--33. IEEE Computer Soc., Los Alamitos, CA, [2020] \copyright 2020.

\bibitem{sharma2023noise}
Asim Sharma and Avah Banerjee.
\newblock Noise-aware token swapping for qubit routing.
\newblock In {\em 2023 IEEE International Conference on Quantum Computing and Engineering (QCE)}, volume~1, pages 82--88. IEEE, 2023.

\bibitem{MR2431751}
Zhizhang Shen and Ke~Qiu.
\newblock On the {W}hitney numbers of the second kind for the star poset: comment on [{E}uropean {J}. {C}ombin. {\bf 11} (1990), no. 3, 277--288; mr1059558] by {F}. {J}. {P}ortier and {T}. {P}. {V}aughan.
\newblock {\em European J. Combin.}, 29(7):1585--1586, 2008.

\bibitem{siraichi2019qubit}
Marcos~Yukio Siraichi, Vin{\'\i}cius Fernandes~dos Santos, Caroline Collange, and Fernando Magno~Quint{\~a}o Pereira.
\newblock Qubit allocation as a combination of subgraph isomorphism and token swapping.
\newblock {\em Proceedings of the ACM on Programming Languages}, 3(OOPSLA):1--29, 2019.

\bibitem{surynek2018finding}
Pavel Surynek.
\newblock Finding optimal solutions to token swapping by conflict-based search and reduction to sat.
\newblock In {\em 2018 IEEE 30th International Conference on Tools with Artificial Intelligence (ICTAI)}, pages 592--599. IEEE, 2018.

\bibitem{surynek2019multi}
Pavel Surynek.
\newblock Multi-agent path finding with generalized conflicts: An experimental study.
\newblock In {\em International Conference on Agents and Artificial Intelligence}, pages 118--142. Springer, 2019.

\bibitem{MR1137822}
Theresa~P. Vaughan.
\newblock Bounds for the rank of a permutation on a tree.
\newblock {\em J. Combin. Math. Combin. Comput.}, 10:65--81, 1991.

\bibitem{MR1705338}
Theresa~P. Vaughan.
\newblock Factoring a permutation on a broom.
\newblock {\em J. Combin. Math. Combin. Comput.}, 30:129--148, 1999.

\bibitem{MR1334632}
Theresa~P. Vaughan and Frederick~J. Portier.
\newblock An algorithm for the factorization of permutations on a tree.
\newblock {\em J. Combin. Math. Combin. Comput.}, 18:11--31, 1995.

\bibitem{MR4594484}
Friedrich Wagner, Andreas B\"{a}rmann, Frauke Liers, and Markus Weissenb\"{a}ck.
\newblock Improving quantum computation by optimized qubit routing.
\newblock {\em J. Optim. Theory Appl.}, 197(3):1161--1194, 2023.

\bibitem{MR3917573}
Katsuhisa Yamanaka, Erik~D. Demaine, Takashi Horiyama, Akitoshi Kawamura, Shin-ichi Nakano, Yoshio Okamoto, Toshiki Saitoh, Akira Suzuki, Ryuhei Uehara, and Takeaki Uno.
\newblock Sequentially swapping colored tokens on graphs.
\newblock {\em J. Graph Algorithms Appl.}, 23(1):3--27, 2019.

\bibitem{MR3349550}
Katsuhisa Yamanaka, Erik~D. Demaine, Takehiro Ito, Jun Kawahara, Masashi Kiyomi, Yoshio Okamoto, Toshiki Saitoh, Akira Suzuki, Kei Uchizawa, and Takeaki Uno.
\newblock Swapping labeled tokens on graphs.
\newblock {\em Theoret. Comput. Sci.}, 586:81--94, 2015.

\bibitem{YAMANAKA:2015}
Katsuhisa Yamanaka, Erik~D. Demaine, Takehiro Ito, Jun Kawahara, Masashi Kiyomi, Yoshio Okamoto, Toshiki Saitoh, Akira Suzuki, Kei Uchizawa, and Takeaki Uno.
\newblock Swapping labeled tokens on graphs.
\newblock {\em Theoretical Computer Science}, 586:81--94, 2015.
\newblock Fun with Algorithms.

\bibitem{yasui2015swapping}
Gaku Yasui, Kouta Abe, Katsuhisa Yamanaka, and Takashi Hirayama.
\newblock Swapping labeled tokens on complete split graphs.
\newblock {\em Inf. Process. Soc. Japan. SIG Tech. Rep}, 14:1--4, 2015.

\bibitem{MR1666061}
Louxin Zhang.
\newblock Optimal bounds for matching routing on trees.
\newblock {\em SIAM J. Discrete Math.}, 12(1):64--77, 1999.

\end{thebibliography}

\appendix

\section{New 4-approximation}

We give a new 4-approximation for \textsc{Token Swapping} on graphs, inspired by the 2-approximation for trees given by \cite{YAMANAKA:2015}.

Given a token swapping instance $K = (G, T, f_1, f_2)$ on $n$ vertices, we observe that the functions $f_1$ and $f_2$ induce a permutation $\pi$ on $[n]$, which maps a vertex $v$ to the target of the token starting on $v$. Moreover, $\pi$ can be decomposed in polynomial time into cycles $\pi = C_1 \circ ... \circ C_k$. The cycle $C_i$ consists of vertices $v_{i, 0}, ..., v_{i, \ell_i-1}$, where the destination for the token starting on $v_{i, j}$ is $v_{i,j+1 \pmod{\ell_i}}$. That is, $f_2(f_1^{-1}(v_{i, j})) = v_{i,j+1 \pmod{\ell_i}}$. Our algorithm proceeds in $k$ phases. In the $i$-th phase, the tokens in $C_i$ are brought to their targets, while all other tokens end the phase on the same vertices they started it on. Below is a formal description.
\begin{enumerate}
    \item For each $i = 1, ..., k$:
    \begin{enumerate}
        \item For each $j = \ell_i-1, \ell_i-2, ..., 1$:
        \begin{enumerate}
            \item Set $t_{i,j-1}$ to be the token currently on $v_{i,j-1}$, set $t_{i,j}$ to be the token currently on $v_{i,j}$, and set $p_{i,j}$ to be a shortest path between them.
            \item Bubble $t_{i,j-1}$ along $p_{i,j}$ until it reaches $v_{i,j}$.
            \item Bubble $t_{i,j}$ in reverse order along $p_{i,j}$ until it reaches $v_{i,j-1}$.
        \end{enumerate}
    \end{enumerate}
\end{enumerate}
\begin{claim} \label{claim:new-algo}
    For $i \in [k]$, for $j' \in \{0, ..., \ell_1 -1\}$, the token which begins the $i$-th phase on $v_{i,j'}$ ends the phase on $v_{i,j'+1}$.
\end{claim}
\begin{proof}
    If $j' \neq \ell_{i-1}$, then during the iteration of the inner loop where $j = j'+1$, the token beginning the iteration on $v_{i,j'}$ is moved during Step ii to $v_{i,j'+1}$. If, during any other iteration of the inner loop, this token is knocked from its current vertex by Step ii, it is moved back by Step iii.

    If $j' = \ell_{i-1}$, then the token which begins the $i$-th phase on $v_{i,j'}$ is set to $t_{i,j}$ during every iteration of the inner loop. In the first iteration, $j = \ell_i-1$, and this token is on $v_{i,\ell_i-1}$. In the inductive case, it is brought from $v_{i,j}$ to $v_{i,j-1}$ by Step iii. of the iteration.
\end{proof}
Moreover, each token which begins the $i$-th phase on a vertex not in $C_i$ ends the phase on the same vertex. If it is displaced by Step ii during a given iteration of the inner loop, then it is moved back to the vertex it was previously on during Step iii. As a consequence, \cref{claim:new-algo} implies that each token that starts on a vertex in $C_i$ is moved to its destination during the $i$-th phase. Next, we will show that the algorithm is indeed a 4-approximation.
\begin{claim}
    Let $OPT(K)$ be length of the shortest swap sequence for $K$, and let $ALG(K)$ be the length of the swap sequence returned by the above algorithm. Then $OPT(K) \leq ALG(K) \leq 4 \cdot OPT(K)$.
\end{claim}
\begin{proof}
    Let $total(K)$ denote the sum of the distances between tokens and their targets in $K$:
    \[ total(K) \coloneq \sum_{t \in T} \text{dist}(f_1(t), f_2(t)).\]
    Then $\frac{1}{2} total(k) \leq OPT(K)$, as each swap can bring at most two tokens closer to their target vertex. We will show that $ALG(K) < 2 \cdot total(K)$.

    Because each vertex is contained in exactly one cycle in $\pi$, we can rewrite $total(K)$ as:
    \begin{align*}
        total(K) &= \sum_{i = 1}^k \sum_{j = 0}^{\ell_i - 1} \text{dist}(f_1(f_2^{-1}(v_{i,j})),v_{i,j}) \\
        &= \sum_{i = 1}^k \sum_{j = 0}^{\ell_i - 1} \text{dist}(v_{i,j-1\pmod{\ell_i}},v_{i,j}).
    \end{align*}
    For each iteration of the inner loop in the swap sequence, there is a unique choice of $i \in [k]$, $j' \in \{ 0, ..., \ell_i - 2\}$, so that $i = i'$ and $j = j'+1$. During this iteration, Step ii performs 
    $\text{dist}(v_{i,j},v_{i,j+1})$ swaps, and Step iii performs $\text{dist}(v_{i,j},v_{i,j+1}) -1$ swaps. Therefore, the total number of swaps in the sequence is:
    \begin{align*}
        ALG(K) &= \sum_{i = 1}^k \sum_{j = 1}^{\ell_i - 1} 2 \cdot \text{dist}(v_{i,j-1},v_{i,j}) - 1 \\
        &< \sum_{i = 1}^k \sum_{j = 0}^{\ell_i - 1} 2 \cdot \text{dist}(v_{i,j-1 \pmod{\ell_i}},v_{i,j}) \\
        &= 2 \cdot total(K),
    \end{align*}
    as desired.
\end{proof}

\end{document}